%% file: main_OutputAgreement.tex
\documentclass[a4paper, 12pt]{article}

\usepackage{amsmath,amssymb,amsfonts,theorem}
\usepackage{graphicx,color}
\usepackage{authblk}
\usepackage{fullpage}
\usepackage{tikz}
\usetikzlibrary{matrix}
\usetikzlibrary{positioning}
\usetikzlibrary{arrows,shapes, shadows, calc}
\usetikzlibrary{spy}
\usepackage{pgfplots}
\pgfplotsset{compat=1.3}

\usepackage{subcaption}

{ \theorembodyfont{\normalfont} 
\newtheorem{example}{Example}

}
\newtheorem{assumption}{Assumption}
\newtheorem{definition}{Definition}
\newtheorem{theorem}{Theorem}

\newtheorem{corollary}{Corollary}
\newtheorem{proposition}{Proposition}

\newenvironment{myassumption}[2][Assumption]{\begin{trivlist}
\item[\hskip \labelsep {\bfseries #1}\hskip \labelsep {\bfseries #2}]}  {\end{trivlist}}

\newcommand{\R}{\mathbb{R}}

\newcommand{\1}{\mathbf{1}} 
\newcommand{\0}{\mathbf{0}}

\newcommand{\dist}{\rm dist}

\newcommand{\dst}{\displaystyle}

\def\be{\begin{equation}}
\def\ee{\end{equation}}
\def\ba{\begin{array}}
\def\ea{\end{array}}
\def\eqa{\begin{eqnarray}}
\def\eqe{\end{eqnarray}}

\def\stopmodif{\color{black}} 

\newcommand{\mc}{\mathcal}

\begin{document}

%
%
%

\title{Dynamic coupling design for nonlinear output agreement and time-varying flow control
}
\author[1]{Mathias B\"{u}rger\thanks{Work supported in part by the Cluster of Excellence in Simulation Technology (EXC301/2) at the University of Stuttgart.}}
\author[2]{Claudio De Persis\thanks{Work supported by the research grants
{\it Efficient Distribution of Green Energy} (Danish Research Council of Strategic Research), {\it Flexiheat} (Ministerie van Economische Zaken, Landbouw en Innovatie), and by a starting grant of the Faculty of Mathematics and Natural Sciences, University of Groningen.}}
\affil[1]{Institute for Systems
Theory and Automatic Control, University of Stuttgart, Pfaffenwaldring 9, 70550 Stuttgart, Germany {\tt mathias.buerger@ist.uni-stuttgart.de}
}
\affil[2]{ITM, Faculty of Mathematics and Natural Sciences, University of Groningen, Nijenborgh 4,  9747 AG Groningen, the Netherlands {\tt c.de.persis@rug.nl}
}

\maketitle 

\begin{abstract}
This paper studies the problem of output agreement in networks of nonlinear dynamical systems under time-varying disturbances, using dynamic diffusive couplings.
Necessary conditions are derived for general networks of nonlinear systems, and these conditions are explicitly interpreted as conditions relating the node dynamics and the network topology. 
For the class of incrementally passive systems, necessary and sufficient conditions for output agreement are derived. 
The approach proposed in the paper lends itself to solve flow control problems in distribution networks.  
As a first case study, the internal model approach is used for designing a controller that achieves an optimal routing and inventory balancing in a dynamic transportation network with storage and time-varying supply and demand. It is in particular shown that the time-varying optimal routing problem can be solved by applying an internal model controller to the dual variables of a certain convex network optimization problem. 
As a second case study, we show that droop-controllers in microgrids have also an interpretation as internal model controllers.
\end{abstract}


\section{Introduction}
%
Output agreement has evolved as one of the most important control objectives in cooperative control. It appears in various contexts, ranging from distributed optimization (\cite{Tsitsiklis1986}), formation control \cite{Olfati-Saber2007} up to oscillator synchronization (\cite{Stan2007}). 
%
%
Over the last years, it has become evident that the internal model principle takes a central role in output agreement problems, see e.g. \cite{Wieland2011}, \cite{Bai2011}, \cite{DePersisCDC2012}, \cite{DePersis2013}. \\
The present paper studies output agreement in networks of heterogeneous nonlinear dynamical systems affected by external disturbances. Conditions on the dynamic couplings (or equivalently design principles for controllers placed on the edges of the network) are derived, that ensure output agreement. We follow here the trail opened in \cite{Pavlov2008} for centralized output regulation and provide necessary and sufficient conditions for the solution of the output agreement problem for the class of \emph{incrementally passive} systems. \\
%
%
%
We propose an approach that is inherently different from other internal model approaches such as \cite{Wieland2011} (see \cite{Wieland2013} and \cite{Isidori2013} for an extension to nonlinear system), where systems without external disturbances are considered. The conceptual idea of \cite{Wieland2011} can be summarized as follows. Each node is augmented with a local controller that contains the model of a reference system, identical for all nodes. The local controllers are designed such that the node dynamics asymptotically track the reference system. The local (``virtual'') copies of the reference system are then synchronized with static diffusive couplings. 
The approach considered in the present paper is inherently different. Most obviously, the objective of this paper is the design of dynamic couplings, rather than the design of local controllers. Furthermore, as external signals are assumed to affect the node dynamics, the assumptions of  \cite{Wieland2011} do not hold (e.g., the controllability of the complete node dynamics is not given) and therefore the approach of \cite{Wieland2011} is not applicable. Incrementally passive systems and disturbance rejection are also dealt with in \cite{DePersisCDC2012}. However, the framework we propose here, inspired by \cite{Pavlov2008},   is completely different and leads to a family of new distinct results that have not been considered in \cite{DePersisCDC2012}. 
Therefore, our results complement the existing approaches and add a new perspective to internal model control for output agreement.

The contributions of this paper are as follows.  
We consider networks of nonlinear systems, interacting according to an undirected network topology. The design objective is to design controllers placed on the edges of the network that achieve output agreement.
We present and discuss necessary conditions for the feasibility of the problem. For the class of linear systems, we provide an interpretation of these conditions, relating the node dynamics and the network topology, that explain the important role of passivity in output agreement problems.
Following this, sufficient conditions for output agreement in networks of incrementally passive systems are provided. We prove that the output agreement problem is feasible if one can find an incrementally passive internal model controller.
A relevant class of nonlinear systems is presented, for which the proposed internal model controller design is always possible. 
To clarify the relation to the existing literature, two special situations are discussed, where either output agreement can be reached with static diffusive couplings or where the disturbances are constant. 
Following the general theoretic discussion, the internal model control design approach is shown to be relevant for different applications. 
First, the problem of \emph{optimal routing control} in distribution systems with time-varying demand is considered, as they appear, e.g., in supply chains (\cite{Alessandri2011}) or data networks (\cite{Moss1983}). Following the internal model control design procedure, routing controllers are designed that achieve a balancing of the inventory levels and an optimal routing of the flow. 
Second, it is shown that \emph{droop-controllers} in microgrids, as, e.g., studied in \cite{SimpsonPorco2013}, turn out to be designed exactly in accordance to the internal model control approach. In view of this, the internal model control approach provides the theoretical framework for the analysis and design of networked systems.

%
The remainder of the paper is organized as follows. The problem formulation and necessary conditions for output agreement are presented in Section \ref{sec.OutputAgreementNec}. Sufficient conditions for output agreement in networks of incrementally passive systems are discussed in Section \ref{sec.Sufficient}. 
A constructive procedure for the design of such controllers for a class of nonlinear systems in presented in Section \ref{subsec.class.nonl}.
In Section \ref{sec.Relation}, the relation to known methods in the literature is formally discussed.
The time-varying optimal distribution problem is presented in Section \ref{sec.OptimalDistribution} and the interpretation of droop-controllers a internal model controllers is provided in Section \ref{sec.Droop}.

\textbf{Notation:}
 The set of (positive) real numbers is denoted by $\mathbb{R}$ ($\mathbb{R}_{\geq}$). 
Given two matrixes $A$ and $B$, the \emph{Kronecker product} is denoted by $A \otimes B$. 
The Moore-Penrose inverse (or pseudo-inverse) of a non-invertible matrix $A$ is denoted by $A^{\dagger}$.
The \emph{range-space} and \emph{null-space} of a matrix $B$ are denoted by $\mc{R}(B)$ and $\mc{N}(B)$, respectively.
A graph $\mc{G}=(V,E)$ is an object consisting of a finite set of nodes, $|V| = n$, and edges, $|E| = m$.
The incidence matrix $B \in \mathbb{R}^{n \times m}$ of the
graph $\mc{G}$ with arbitrary orientation, is a $\{0, \pm 1\}$ matrix with $[B]_{ik}$ having value
`+1' if node $i$ is the initial node of edge $k$, `-1' if it is the terminal
node, and `0' otherwise.

\section{Problem formulation and necessary conditions} \label{sec.OutputAgreementNec}

We consider a network of dynamical systems defined on a connected, undirected graph $\mc{G}=(V,E)$. Each node represents a nonlinear system 
\be\label{nonl.systems.0}\ba{rcll}
\dot x_i &=& f_i(x_i, u_i, w_i)\\
y_i &=& h_i(x_i, w_i), & i=1,2,\ldots, n,
\ea\ee
where $x_i\in \R^{r_i}$ is the state, and $u_i , y_i\in \R^{p}$ are the input and output, respectively. Each system \eqref{nonl.systems.0} is driven by the time-varying signal $w_i \in \R^{q_i}$, representing, e.g.,~a disturbance or reference.
We assume that the exogenous signals $w_i$ are generated by systems of the form 
\begin{align} \label{sys.Exo}
\dot w_i =s_i(w_i), \quad w_i(0) \in \mc{W}_i,
\end{align}
where $\mc{W}_i$ is a set whose properties are specified below.
\begin{assumption}
The vector field $s_{i}(w_{i})$ satisfies for all $w_i, w_i'$ the inequality
\begin{align}\label{eqn.incrementalProp}
(w_i-w_i')^T (s_i(w_i)- s_i(w_i'))\le 0. 
\end{align}
\end{assumption}
This is going to be a standing assumption in this paper. As an example, consider the linear function with skew-symmetric matrix
$
s_i(w_i)= S_i w_i,\; S_i^T+S_i=0. $ \\
We stack together the signals $w_i$, for $i=1,2,\ldots, n$, and obtain  the vector $w \in \mathbb{R}^{q}$, which satisfies the equation $\dot w =s(w)$. In what follows, whenever we refer to the solutions of  $\dot w =s(w)$, we assume that the initial condition is chosen in a compact set $\mc{ W}=\mc{W}_1\times \ldots \times \mc{W}_n$.  
The set $\mc{W}$ is assumed to be forward invariant for the system  $\dot w =s(w)$.
Similarly, let $x$, $u$, and $y$ be the stacked vectors of $x_{i}, u_{i},$ and $y_{i}$, respectively. Using this notation, the totality of all systems is given by
\be\label{nonl.systems.tot}
\ba{rcl}
\dot{w} &=& s(w) \\
\dot x &=& f(x, u, w)\\
y &=& h(x,w)
\ea
\ee
with state space $\mc{W} \times \mc{X}$ and $\mc{ X}$ a  compact  subset of $\R^{r_1}\times \ldots \times \R^{r_n}$. \\
The control objective is to reach output agreement of all nodes in the network, independent of the exact representation of the time-varying external signals.
Therefore, between any pair of neighboring nodes, i.e., on any edge of $\mc{G}$, a dynamic controller will be placed, taking the form
\be\label{nonl.im}
\ba{rcl}
\dot \xi_k &=& F_{k}(\xi_k, v_k)\\
\lambda_k &=& H_{k}(\xi_k, v_{k}), \quad k=1,2,\ldots, m,
\ea
\ee
with state $\xi_k\in \R^{\nu_k}$ and input $v_k\in \R^p$.
When stacked together, the controllers (\ref{nonl.im}) give raise to the overall controller
\be\label{im.big}
\ba{rcl}
\dot \xi &=& F(\xi,v)\\
\lambda &=& H(\xi, v),
\ea\ee
where $\xi \in \Xi$, a  compact  subset of $\R^{\nu_1}\times\ldots\times \R^{\nu_m}$. \\
Throughout the paper the following interconnection structure between the plants, placed on the nodes of $\mc{G}$, and the controllers, placed on the edges of $\mc{G}$, is considered. A controller (\ref{nonl.im}), associated with edge $k$ connecting nodes $i,j$, has access to the relative outputs $y_i-y_j$. In vector notation, the relative outputs of the systems are
\be\label{interc0}
z= (B\otimes I_p)^T y.
\ee
The controllers are then driven by the systems via the interconnection condition
\be\label{interc1}
v=-z,
\ee
where $v$ are the stacked inputs of the controllers.
Additionally, the output of the controllers influence the incident systems via the interconnection\footnote{
The interconnection structure \eqref{interc0}, \eqref{interc2} naturally represents a canonical structure for distributed control laws. This structure is often considered in the context of passivity-based cooperative control, see e.g., \cite{Arcak2007}, \cite{Bai2011},  \cite{Schaft2012}, \cite{DePersis2012a}, \cite{Burger2013}.}
\be\label{interc2}
u = (B \otimes I_p)\lambda. 
\ee 
Due to this interconnection structure the dynamics on the network can be represented as a closed-loop dynamics as illustrated in Figure \ref{fig.Structure1}.
\begin{figure}[t]
\begin{center}
\scalebox{0.6}{\input{BlockDiagram01}}
\end{center}
\caption{Structure of the internal model control scheme.} \label{fig.Structure1}
\end{figure}
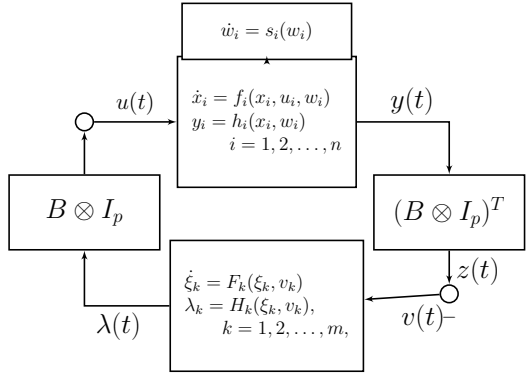
We are now ready to formally introduce the output agreement problem. 
\begin{definition}[Output Agreement Problem] \label{def.OutputAgreement}
The \emph{output agreement problem} is solvable for the process (\ref{nonl.systems.tot}) under the interconnection relations (\ref{interc0}), (\ref{interc1}), (\ref{interc2}), if there exists controllers  (\ref{im.big}), such that every solution $(w(t), x(t), \xi(t))$ originating from $\mc{W}\times \mc{X}\times \Xi$ is bounded and satisfies 
$\lim_{t\to \infty}~(B^T\otimes~I_p)~y(t)~=~\mathbf{0}$.  
\end{definition}

\subsection{Necessary Conditions}\label{subs.nc}
To start the discussion, we first investigate the \emph{necessary conditions} for the output agreement problem to be solvable. 
To this purpose, we strengthen the requirement on the convergence of the regulation error to the origin, requiring that $\lim_{t\to \infty}~(B^T\otimes~I_p)~y(t)~=~\mathbf{0}$ {\em uniformly} in the initial condition (\cite{AI-CB-ARC08}).   
The closed-loop system 
(\ref{nonl.systems.tot}), (\ref{im.big}), (\ref{interc0}), (\ref{interc1}), (\ref{interc2})
can be written as
\be\label{overall}\ba{rcll} 
\dot w &=& s(w)\\
\dot x &=& f (x, (B\otimes I_p) H(\xi), w)\\
\dot \xi &=& F(\xi,-(B\otimes I_p)^T h (x , w)). 
\ea\ee

\begin{definition}[{\small $\omega$-limit set}]
The $\omega$-limit set  $\Omega(\mc{W}\times \mc{X}\times \Xi)$ is the set of points $(w, x,\xi)$ for which there exists a sequence of pairs $(t_k, (w_k, x_k,\xi_k))$ with $t_k \to +\infty$ and $(w_k, x_k,\xi_k)\in \mc{W}\times \mc{X}\times \Xi$ such that $\varphi(t_k, (w_k, x_k,\xi_k))\to (w, x,\xi)$ as $k\to +\infty$, where $\varphi(\cdot, \cdot)$ is  the flow of (\ref{overall}).
\end{definition} 
If the output agreement problem is solvable, then  the $\omega$-limit set
$\Omega(\mc{W}\times \mc{X}\times \Xi)$ is nonempty, compact, invariant and uniformly attracts $\mc{W}\times \mc{X}\times \Xi$ under the flow of (\ref{overall}). 
Furthermore, the $\omega$-limit set must satisfy
\begin{align*}
 \Omega&(\mc{W}\times \mc{X}\times \Xi) \subseteq  \bigl\{(w,x,\xi) \in \mc{W}\times \mc{X} \times \Xi : (B\otimes I_p)^T h (x, w) = \mathbf{0} \bigr\}.
 \end{align*}
This set is the graph of a map defined on the whole $\mc{W}$ and is invariant for the closed-loop system. By the invariance, for any solution $w$ of the exosystem originating from $\mc{W}$, there exists  
$(x^w, u^w, \xi^w)$ such that 
\be\label{regulator.equations}\ba{rcl}
\dot x^w &=& f(x^w, u^w,w)\\
\0&=& (B\otimes I_p)^T h (x^w, w)
\ea\ee
and 
\be \label{controller.constraints_O}
\ba{rcl}
\dot \xi^w &=&  F(\xi^w,0)\\
u^w &=& (B\otimes I_{p}) H(\xi^w, 0). 
\ea\ee
\begin{proposition}
If the output agreement problem is solvable, then, 
for every $w$ solution to $\dot w =s(w)$ originating in $\mc{W}$, 
there must exist solutions $(x^w, u^w, \xi^w)$ such that  the equations (\ref{regulator.equations}), (\ref{controller.constraints_O}) are satisfied. 
\end{proposition}
In a controller-independent form, the constraints  \eqref{regulator.equations} and \eqref{controller.constraints_O} require that there exists $(x^{w},u^{w})$ satisfying
\begin{align}
\begin{split}  \label{eqn.regulator.equations.controller.independent}
\dot x^w &= f(x^w, u^w,w), \; y^{w} = h (x^w, w) \\
u^{w} &\in \mc{R} (B\otimes I_{p}), \quad y^{w} \in  \mc{N}(B^{T}\otimes I_p),
\end{split}
\end{align} 
where $u^{w} \in \mc{R}(B\otimes I_{p})$ denotes that at every time $t$ the vector $u^{w}(t)$ is contained in the respective vector space.
Let in the following  $u^{w}$ be a solution to \eqref{eqn.regulator.equations.controller.independent}, and $\lambda^{w}_{p}$ be a trajectory satisfying $u^{w} = (B\otimes I_{q})\lambda^{w}_{p}$. The trajectory $\lambda^{w}_{p}$ is uniquely defined if and only if the graph $\mc{G}$ has no cycles. Otherwise, the matrix $B$ has a nontrivial nullspace, see \cite{Godsil2001}.  
In the most general form, the existence of a feedforward controller is equivalent to the constraint that there exists an integer $d$ and maps  $\tau: \mc{W} \mapsto \mathbb{R}^{d}$, $\phi : \mathbb{R}^{d} \mapsto \mathbb{R}^{d}$ and $\psi: \mathbb{R}^{d} \mapsto \mathbb{R}^{mp}$ satisfying 
\begin{align}\label{embedding}
\begin{split}
\dst\frac{\partial \tau}{\partial w} s(w) &= \phi(\tau(w))\\
\lambda^w_{p}   + \lambda^{w}_{0}&= \psi(\tau(w)), \; \lambda^{w}_{0} \in \mc{N}\bigl( B\otimes I_{p} \bigr).
\end{split}
\end{align}
Note that there might be an infinite number of possible controllers that can generate the desired steady state input $u^{w}$. 
If the constraint \eqref{embedding} holds, the system
\be\label{internal.model}\ba{rcl}
\dot \eta &=& \phi(\eta), \; \eta \in \mathbb{R}^{d} \\
\lambda &=& \psi(\eta)
\ea\ee
has the property that if $\eta_0= \tau(w(0))$, then the solution $\eta(t)$ to (\ref{internal.model}) starting from $\eta_0$ is  such that $(B \otimes I_{p})\lambda(t) = u^{w}(t)$ for all $t\ge 0$. 
We denote by $\eta^w$ such a solution  to (\ref{internal.model}) such that  $(B \otimes I_{p}) \psi(\eta^w(t)) = u^{w}(t)$ for all $t\ge 0$. We then let  $\lambda^w(t):= \psi(\eta^w(t))$. Here $\lambda^w(t)$ is one of the infinite many realizations of the map $\lambda_p^w(t)+\lambda_0^w(t)$, with $\lambda^{w}_{0} \in \mc{N}\bigl( B\otimes I_{p} \bigr)$.

To design a controller that decomposes into controllers on the edges of $\mc{G}$,
 we introduce a vector $\eta_{k} \in \mathbb{R}^{d}$ for each edge $k = 1,\ldots,m,$ and denote with $\psi_{k}$ the entries of the vector valued function $\psi$ corresponding to the edge $k$. Each edge is now assigned a controller of the form
\be\label{internal.model.k1}\ba{rcll}
\dot \eta_k &=& \phi(\eta_k),\; \lambda_k = \psi_k(\eta_k), & k=1,2,\ldots, m.  
\ea\ee
With the stacked vector $\eta = [\eta_{1}^{T}, \ldots, \eta_{m}^{T}]^{T}$, and vector valued functions 
$\bar{\phi}(\eta) = [ \phi(\eta_{1}), \ldots,  \phi(\eta_{m})]^{T}$, $\bar{\psi}(\eta) = [ \psi_{1}(\eta_{1}), \ldots, \psi_{m}(\eta_{m})]^{T}$,
the overall controller \eqref{im.big} is
\begin{align} \label{sys.GlobalIM}
\begin{split}
\dot{\eta} &= \bar{\phi}(\eta) \\
\lambda &= \bar{\psi}(\eta).
\end{split}
\end{align}
If the initial condition is chosen as $\eta_0= I_{m} \otimes \tau(w(0))$ then the solution $\eta(t)$ to (\ref{internal.model}) starting from $\eta_0$ is  such that $\lambda(t)=\lambda^w(t)$ for all $t\ge 0$.

\subsection{Discussion:  The Regulator Equations}
The necessary conditions \eqref{regulator.equations} and \eqref{controller.constraints_O} are a weaker form of the regulator equations of \cite{Isidori1990}. If the systems \eqref{nonl.systems.0} are such that for each given exogeneous input $w(t)$ there exists a \emph{unique} steady state response, and the $\omega$-limit set can be expressed as
$
\Omega(\mc{W}\times \mc{X}\times \Xi)=\{(w, x, \xi)\,:\,  x=\pi(w), \xi=\pi_c(w)\}, 
$
then $x^w= \pi(w)$ and the regulator equations (\ref{regulator.equations}) express the existence of an invariant manifold where the ``regulation error" $(B^T\otimes I_p)y$ is identically zero provided that the control input $u^w$ is applied. Furthermore, (\ref{controller.constraints_O}) express the existence of a controller able to  provide $u^w$.  In this case, (\ref{regulator.equations}), (\ref{controller.constraints_O}) take the familiar expressions, see e.g. \cite{Isidori1990}:
\begin{align}
\begin{split} \label{eqn.FBI_Sys}
\ba{rcl}
\dst\frac{\partial \pi}{\partial w} s(w) &=&  f(\pi(w), (B\otimes I_p) H(\pi_c(w)),w) \\
0&=& (B\otimes I_p)^T h (\pi(w), w)
\ea
\end{split}
\end{align}
and 
\begin{align}
\begin{split} \label{eqn.FBI_Cont}
\ba{rcl}
\dst\frac{\partial \pi_c}{\partial w} s(w) &=&  F(\pi_c(w),0)\\
\ea
\end{split}
\end{align}
However, there is a substantial structural difference between the output agreement problem considered here and output regulation problems, that can be best seen for linear dynamical systems.
 Suppose each system \eqref{nonl.systems.0} is of the form
\begin{align}
\begin{split} \label{eqn.linsys}
\dot{x}_{i} &= A_{i}x_{i} + G_{i}u_{i} + P_{i}w_{i} \\
y_{i} &= C_{i}x_{i},
\end{split}
\end{align}
with a linear exosystem $\dot{w}_{i} = S_{i}w_{i}$.
Let in the following $\bar{A} = \mathrm{block.diag}(A_{1},\ldots,A_{n})$, $\bar{G} = \mathrm{block.diag}(G_{1},\ldots,$ $G_{n})$, $\bar{P} = \mathrm{block.diag}(P_{1},\ldots,P_{n})$ and $\bar{C} = \mathrm{block.diag}(C_{1},$  $\ldots,C_{n})$. 
The exosystems are stacked into the dynamics $w = \bar{S}w$, with $\bar{S} = \mathrm{block.diag}(S_{1},\ldots,S_{n})$.  
The classical result of \cite{Francis1976} states that one can take $x^{w} = \Pi w$ and $\lambda^{w} = \Gamma w$ such that the regulator equations \eqref{eqn.FBI_Sys} take the form of Sylvester equations
\begin{align}
\begin{split} \label{eqn.FrancisEquations}
 &\Pi \bar{S} = \bar{A} \Pi + \bar{G}(B \otimes I_{p}) \Gamma + \bar{P} \\
 &(B^{T} \otimes I_{p})\bar{C} \Pi = 0.
\end{split}
\end{align}
Under controllability and observability assumptions, feasibility of \eqref{eqn.FrancisEquations} is necessary and sufficient for output regulation of linear systems. We will see next, that due to the networked structure of the considered problems the assumptions fail to hold, although the output agreement problem is solvable (as we show in the next sections).
First note that the regulator equations \eqref{eqn.FrancisEquations} have a solution if and only if
\begin{align} \label{eqref.FrancisCond}
\mathrm{rank} \Bigl(  \begin{bmatrix}    \bar{A} - s I_{r} & \bar{G}(B \otimes I_{p}) \\
(B\otimes I_{p})^{T}\bar{C} & \0_{np \times np} \end{bmatrix} \Bigr)  = \# \mathrm{rows},\
\end{align}
for all $s \in \sigma(\bar{S})$, where $r=\sum_{i=1}^{n}r_{i}$ and $\sigma(\bar{S})$ is the spectrum of $\bar{S}$.
The condition states that no pole of the stacked exosystem is a transmission zero of the system from input $\lambda$ to output $z = (B\otimes I_{p})^{T}y$. 
To focus the discussion on the impact of the constraints resulting from the network, we impose the following assumption:
\begin{assumption} \label{ass.NonResonance}
For each system $i\in \{1,\ldots,n\}$
\begin{align*} 
\mathrm{rank} \Bigl(  \begin{bmatrix}    A_{i} - s I_{r_{i}} & G_{i} \\
C_{i} & \0_{p \times p} \end{bmatrix} \Bigr)  = r_{i} + p,\; \forall \; s \in \sigma(\bar{S}).
\end{align*}
\end{assumption}

The important observation is that the rank condition can be violated due to the networked structure of the problem. We summarize this in the result below.
\begin{proposition} \label{prop.FailureRE}
Suppose Assumption \ref{ass.NonResonance} holds. The rank condition \eqref{eqref.FrancisCond} is violated 
if either of the following holds:
\begin{enumerate}
\item \label{LinItem1} $\mc{G}$ contains a cycle;
\item \label{LinItem2} $\mc{R} \Bigl( \bar{H}(s) (B \otimes I_{p}) \Bigr) \cap \mc{N}\Bigl( (B^{T} \otimes I_{p}) \Bigr) \neq \{ \0 \}$
for some $s \in \sigma(\bar{S})$, where $\bar{H}(s) = \bar{C}(sI - \bar{A})^{-1}\bar{G}$.   
\end{enumerate} 
Moroever, the conditions are necessary provided that for all $s\in \sigma(\bar S)$, $s\not \in \sigma(\bar A)$. 
\end{proposition}
The proof is presented in the appendix. 
The first conditions shows that the regulator equations \eqref{eqn.FrancisEquations} have no solution if graph contains cycles 
or if the transfer functions of the dynamical systems ``rotate'' $\mc{R} (B \otimes I_{p})$ in such a way that it intersects nontrivially
its orthogonal space  $\mc{N} (B^{T} \otimes I_{p})$.
The previous result gives an intuition about a class of systems for which the output agreement problem is feasible.
\begin{corollary} \label{cor.SPR}
Assume Assumption \ref{ass.NonResonance} holds and $\mc{G}$ contains no cycles. Suppose furthermore that all eigenvalues of $\bar{S}$ have zero real part.
Then the equations \eqref{eqn.FrancisEquations} are feasible if $\bar{H}(s)$ is \emph{strictly positive real}.\footnote{We refer to \cite[Def. 6.4]{Khalil2002} for the definition of a strictly positive real transfer function.}
\end{corollary}
The proof is presented in the appendix. The result suggests that passivity takes an outstanding role in the output agreement problem. 

\section{Output agreement under time-varying disturbances} \label{sec.Sufficient}

In this section we highlight sufficient conditions that lead to a solution of the problem for a special class of systems, namely \emph{incrementally passive} systems. Our approach follows the line of \cite{Pavlov2008}, where the following notion of a regular storage function was introduced.
\begin{definition}[\cite{Pavlov2008}]
A storage function $V(t,x,x')$ is called regular if for any sequence $(t_{k},x_{k},x_{k}')$, $k=1,2,\ldots$, such that $x'_{k}$ is bounded, $t_{k}$ tends to infinity, and $|x_{k}| \rightarrow \infty$, it holds that $V(t_{k},x_{k},x_{k}') \rightarrow \infty$, as $k \rightarrow \infty$.
\end{definition}
The dissipativity characterization of incremental passivity provided in \cite{Pavlov2008} is as follows.
\begin{definition}\label{pip}			
The system \eqref{nonl.systems.0} is said to be \emph{incrementally passive} if there exists a $C^{1}$ regular storage function $V_i:\mathbb{R}_{\geq 0} \times \mathbb{R}^{r_{i}} \times \mathbb{R}^{r_{i}} \to \mathbb{R}_{\geq 0}$ such that for any two inputs $u_{i}, u_{i}'$ and any two solutions $x_{i}$,$x_{i}'$, corresponding to these inputs, the respective outputs $y_{i}$, $y_{i}'$ satisfy
\be\label{di.x}\ba{l}
\dst\frac{\partial V_i}{\partial t}+\dst\frac{\partial V_i}{\partial x_i} f_i(x_i, u_i , w_i) +
\dst\frac{\partial V_i}{\partial x'_i} f_i(x_i', u_i ', w_i) \le (y_i-y_i')^T (u_i -u_i ').
\ea
\ee
\end{definition}

\setcounter{example}{0} 
\begin{example} Linear  systems of the form \eqref{eqn.linsys} that are passive from the input $u_i $ to the output $y_i$ are also incrementally passive, with $V_i=\frac{1}{2} (x_i-x_i')^T Q_i (x_i-x_i')$ and  $Q_i=Q_i^T>0$ the matrix such that $A_i^T Q_i +Q_i A_i\le 0$ and $Q_i G_i =C_i^T$.  
\end{example}

\begin{example}\label{ex2}
Nonlinear systems of the form
\be\label{sys.ex2}\ba{rcl}
\dot x_i &=& f_i(x_i) +G_i u_i  +P_i w_i\\
y_i &=& C_i x_i 
\ea\ee
with $f_i(x_i)=\nabla F_i(x_i)$, $F_i(x_i)$ twice continuously differentiable and concave, and $G_i=C_i^T$ are incrementally passive. In fact, by concavity of $F_i(x_i)$, $(x_i-x_i')^T ( f_i(x_i)- f_i(x_i'))\le 0$, and $V_i=\frac{1}{2} (x_i-x_i')^T (x_i-x_i')$ is the incremental storage function. 
\end{example}
For the sake of brevity, we will in the following sometimes write $\dot{V}_{i}$ for the directional derivative $\frac{\partial V_i}{\partial t}+\frac{\partial V_i}{\partial x_i} f_i(x_i, u_i , w_i) +
\frac{\partial V_i}{\partial x'_i} f_i(x_i', u_i ', w_i)$.  Incremental passivity can also be defined for static nonlinear systems. A static system $y_{i} = h_{i}(u_{i},t)$ is said to be incrementally passive if it satisfies the monotonicity condition
\begin{align}\label{static.inc.pass}
(h_{i}(u_{i},t) - h_{i}(u_{i}',t))^{T}(u_{i} - u_{i}') \geq 0,
\end{align} 
for all input pairs $u_{i}, u_{i}'$ and all times $t \geq 0$.

In the previous section, it was shown that the controllers at the edge have to take the form  \eqref{internal.model.k1}.
Now, they must be completed by considering additional control inputs that guarantee the achievement of the steady state. While we require the internal model to be identical for all edges, i.e., $\phi(\eta_k)$, the augmented systems might be different. 
Then, the controllers \eqref{internal.model.k1} modify as 
\be\label{internal.model.k.plus.stab}\ba{rcll}
\dot \eta_k &=& \phi_{k}(\eta_k,v_k) \\
\lambda_k &=& \psi_k(\eta_k), & k=1,2,\ldots, m, 
\ea\ee
where all controllers reduce to the common internal model if no external forcing is applied, i.e., $\phi_k(\eta_k,0) =\phi(\eta_k) $. The controller is then said to have the \emph{internal model property}.
The following is the main standing assumption that the controllers must satisfy. 

\begin{assumption}\label{a.ip.im}
For each $k=1,2,\ldots, m$, there exists regular storage functions $W_k(\eta_k, \eta_k')$, with $W_k:\R^{q_k}\times \R^{q_k}\to \R_+$ such that
\begin{align}\label{di.eta} 
\begin{split}
 \dst\frac{\partial W_k}{\partial \eta_k} \phi_k(\eta_k, v_k) + 
 \dst\frac{\partial W_k}{\partial \eta'_k} \phi_k(\eta'_k, v'_k)  \le
 (\lambda_k-\lambda_k')^T (v_k-v_k').
 \end{split}
\end{align}
\end{assumption}

It is in general difficult to design the incrementally passive 
 controllers above. An  important example  when the design is possible is when the feedforward control input is linear, that is (\ref{embedding}) is satisfied with $\tau= {\rm Id}$, $\phi=s$ and $\psi$ is a linear function of its argument. In this case, we let 
$
\phi_k(\eta_k, 0)=s(\eta_k), \quad \psi_k(\eta_k)= M_k\eta_k
$
and define 
\begin{align} \label{eqn.IncPassiveIM}
\phi_k(\eta_k, v_k)=s(\eta_k)+M_k^T v_k.
\end{align}
Then, by definition of $s$ as the gradient of a concave function, the storage function $W_k(\eta_k, \eta_k')=\frac{1}{2} (\eta_k-\eta_k')^T (\eta_k-\eta_k')$ satisfies
\begin{align}
\begin{split} 
 \dst\frac{\partial W_k}{\partial \eta_k} \phi_k(\eta_k, v_k) +
 \dst\frac{\partial W_k}{\partial \eta'_k} \phi_k(\eta'_k, v'_k)
&= 
 (\eta_k-\eta_k')^T (s(w_k)- s(w_k)') 
 +(\eta_k-\eta_k')^T M_k^T (v_k-v_k') \\
& \le   (\psi_k(\eta_k)-\psi_k(\eta_k')) (v_k-v_k'), 
\end{split}
\end{align}
that is (\ref{di.eta}).

We state below the main result of the section that, while  extending to networked systems the results 
of \cite{Pavlov2008}, provides a solution to the output agreement problem in the presence of time-varying disturbances.

\begin{theorem}\label{th1}
Consider the network $\mc{G}$ with dynamics on the nodes  \eqref{nonl.systems.tot}.
Suppose all exosystems satisfy  \eqref{eqn.incrementalProp}, the regulator equations  \eqref{regulator.equations} hold, and all node dynamics are incrementally passive.
Consider the controllers
\be\label{controller.nu}
\ba{rcl}
\dot \eta &=& \bar{\phi}(\eta,v)\\
\lambda &=& \bar{\psi}(\eta) + \nu
\ea\ee
where $\bar{\phi}$ and $\bar{\psi}$ are the stacked functions of $\phi_{k}(\eta_{k},v_{k})$ and $\psi_{k}(\eta_{k}$), and $\nu$ is an additional input to be designed. 
Suppose the controllers have the internal model property and satisfy Assumption \ref{a.ip.im}. 
Then, the controller \eqref{controller.nu} with the interconnection structure
\be\label{int.constr}\ba{ll}
u=(B\otimes I_p) \lambda, \quad v=-(B^T\otimes I_p) y. \\ 
\ea\ee
and $\nu := v = -(B^T\otimes I_p) y$ solves the output agreement problem, that is every solution starting from $\mc{W}\times \mc{X}\times \Xi$ is bounded and 
\[
 \lim_{t\to+\infty} (B\otimes I_p)^T y(t) =\mathbf{0}.  
\]
\end{theorem} 
\begin{proof}
By the incremental passivity property of  the $x$ subsystem in \eqref{nonl.systems.tot}  and (\ref{regulator.equations}), it is true that 
\begin{align*}
\dst\frac{\partial V}{\partial t}+\dst\frac{\partial V}{\partial x} f(x, u , w) +
\dst\frac{\partial V}{\partial x^w} f(x^w, u^w, w) 
\le (y-y^w)^T (u -u^w),
\end{align*}
where $V=\sum_i V_i$. Similarly by Assumption \ref{a.ip.im}, the system  (\ref{controller.nu}) satisfies
\[
 \dst\frac{\partial W}{\partial \eta} \bar{\phi}(\eta, v) +
 \dst\frac{\partial W}{\partial \eta^w} \bar{\phi}(\eta^w) \le
 (\lambda-\lambda^w)^T v - (\nu - \nu^{w})v,
\]
with $W= \sum_k W_k$ and $\bar{\phi}(\eta^{w}) = I_{m} \otimes \phi(\eta^{w})$.
 Bearing in mind the interconnection constraints
$u=(B\otimes I_p)  \lambda$,  $u^w=(B\otimes I_p) \lambda^w,$ and $v=-(B^T\otimes I_p) y,$  
and letting $U((x,x^w), (\eta, \eta^w))=V(x,x^w)+W(\eta, \eta^w)$ we obtain 
\begin{align*}
&\dot U((x,x^w), (\eta, \eta^w)) := \dot{V}(x,x^w) + \dot{W}(\eta, \eta^w)\\
& \leq (y-y^w)^T (u -u^w) +(\lambda-\lambda^w)^T v - (\nu-\nu^{w})^{T}v \\
& = (y-y^w)^T (B \otimes I_p)(\lambda-\lambda^w) \\
&-(\lambda-\lambda^w)^T (B^T \otimes I_p)y+(\nu-\nu^{w})^T (B^T \otimes I_p)y.
\end{align*}
By definition of output agreeement, $(B \otimes I_p)^T y^w=0$ and the previous equality becomes
\begin{align}
\begin{split} \label{eqn.LyapDer}
\dot U((x,x^w), (\eta, \eta^w)) &\leq  
 \nu^T (B^T \otimes I_p)y
=  -||(B^T \otimes I_p)y||^2= -z^T z,
\end{split}
\end{align}
by definition of $\nu=-z$ and $\nu^{w} = 0$. Since $U$ is non-negative and non-increasing, then $U(t)$ is bounded. As $x^w, \eta^w$ are bounded\footnote{By definition, $(w,x^w, \eta^w)$ belongs to the $\omega$-limit set, which is compact. Hence, $x^w, \eta^w$ are bounded.}  and $U$ is regular, then $x,\eta$ are bounded as well. Hence the solutions exist for all $t$.  
Integrating the latter inequality we obtain 
\[
\dst\int_0^{+\infty} z^T(s) z(s) ds \le U(0). 
\]
By Barbalat's lemma, if one proves that $\frac{d}{dt} z^T(t) z(t)$ is bounded then one can conclude that $z^T(t) z(t)  \to 0$. Now, $z(t)=(B^T\otimes I_p) y=(B^T\otimes I_p) h(x,w)$ is bounded because $x,w$ are bounded. If $h$ is continuously differentiable and $\dot x, \dot w$ are bounded, then $\dot z$ is bounded and one can infer that $\frac{d}{dt} z^T(t) z(t)$ is bounded. By assumption, $w$ is the solution of $\dot w=s(w)$ starting from a forward invariant compact set. Hence, both $w$ and $\dot w$ are bounded. On the other hand, $\dot x$ satisfies
\[
\dot x = f(x, (B\otimes I_p)\bar{\psi}(\eta)-z, w)
\]
which proves that  it is bounded because $x, \eta, z$ were proven to be bounded, while $w$ is bounded by assumption. Therefore, $\dot x, \dot w$ are bounded and this implies that $\frac{d}{dt} z^T(t) z(t)$ is bounded. Then by Barbalat's Lemma we have $\lim_{t\to+\infty} z(t)=\mathbf{0}$ as claimed. 
\end{proof}

The result still holds true, if any of the dynamical systems on the nodes or on the edges, is replaced by a static incrementally passive systems. 
As a matter of fact, denoting by $\bar I\subseteq\{1,2,\ldots, n\}$ the subset of indices corresponding to {\em dynamic} incrementally passive systems, it is enough to replace  the Lyapunov function $V$ with $\sum_{i\in \bar I} V_i$. Then, exploiting (\ref{static.inc.pass}),  one can still prove that 
(\ref{eqn.LyapDer}) holds. This proves that the states $x_i$, $i\in \bar I$, $\eta$ are bounded. Notice that the outputs of the static nonlinearities $y_i = h_i(u_i, t)$ are bounded provided that $h_i(u_i, \cdot)$ is a bounded function for every $u_i\in \R^p$. In fact, boundedness of the controller state $\eta$ implies boundedness of $u_i$  since   $u_i =\sum_k b_{ik} \lambda_k(\eta_k)$.
Furthermore, the assumptions on the interconnections can be weakened, if stronger assumptions on the node dynamics are imposed.

\begin{corollary} \label{coro1}
Let all assumptions of Theorem \ref{th1} hold, but assume furthermore that all node dynamics are \emph{output strictly incrementally passive}, that is, 
there exists a $C^{1}$ regular storage function $V_i$, and a positive definite function $\rho_{i}: \mathbb{R}^{p} \rightarrow \mathbb{R}$, such that for any two inputs $u_{i}, u_{i}'$ and corresponding outputs $y_{i}$, $y_{i}'$ 
\be\label{eqn.OSIP}\ba{l}
\dot{V} \leq -\rho_{i}(y_{i}-y_{i}') +  (y_i-y_i')^T (u_i -u_i ').
\ea
\ee
 Then the output agreement problem is feasible with the interconnection \eqref{int.constr} and $\nu = 0.$
\end{corollary}

\begin{proof}
Consider the storage function used in proof of Theorem \ref{th1}, i.e., $U((x,x^{w}),(\eta,\eta^{w})) = V(x,x^{w}) + W(\eta,\eta^{w})$. After repeating the steps of the proof of Theorem \ref{th1}, but using the output strict passivity property and setting $\nu = 0$, \eqref{eqn.LyapDer} is now replaced by
$
\dot{U}((x,x^{w}),(\eta,\eta^{w}))  \leq -\sum_{i=1}^{n}\rho_{i}(y_{i} - y_{i}^{w}), 
$
where $y_{i}^{w}=y_{j}^{w}$ for all $i\ne j$. 
Now, we can proceed as in the proof of Theorem \ref{th1}. 
\end{proof}

\section{Output agreement for a class of nonlinear systems}\label{subsec.class.nonl}

We propose now a fairly large class of nonlinear systems for which the sufficient conditions of Theorem \ref{th1} are satisfied. 
Consider the systems introduced in Example \ref{ex2}, namely  
\be\label{sys.ex2.simpl}\ba{rcll}
\dot w_i &=& s_i(w_i)\\ 
\dot x_i &=& f_0(x_i) + G u_i  +P_i w_i\\
y_i &=& C x_i & i=1,2,\ldots, n
\ea\ee
where, compared with (\ref{sys.ex2}), we have chosen the systems to have the same dynamics, i.e.~$f_i(x_i)=f_0(x_i)$ for all $i=1,2,\ldots, n$, and we have set  $G_i=C_i^T=G=C^T$. Assuming that the dynamics of the systems are the same  facilitate the design of  incrementally passive distributed controllers, as we see in the proof below. 
\begin{proposition}\label{prop.IndenticalSys}
Consider systems (\ref{sys.ex2.simpl}), where $f_0=\nabla F$ and $F$ is a twice continuously differentiable and concave map, $G=C^T$ and full column rank matrix, and the maps $s_i$ satisfy (\ref{eqn.incrementalProp}). Moreover, assume that $\mathcal{R}(P_i)\subseteq \mathcal{R}(G)$ for $i=1,2,\ldots,n$. Given the vector of disturbances $w=(w_1, \ldots, w_n)$, assume there exists a bounded solution $x_\ast$ to the system
\be\label{pre.re}
\dot x=f_0(x) +\dst\frac{\sum_{i=1}^n P_i w_i}{n}.
\ee
Then: 
\begin{enumerate}
\item 
there exists a bounded solution $x^w, u^w$ to the regulator equations (\ref{regulator.equations});
\item there exists  controllers at the edges of the form 
\be\label{ex2.contr}
\ba{rcll}
\dot{\eta}_{k} &=& s(\eta_{k}) + H_{k}v_{k} & \\
\lambda_{k} &=& H_{k}^{T}\eta_{k} +  \nu_{k}, & k=1,2,\ldots, m
\ea
\ee
such that output agreement problem is solved for the systems (\ref{sys.ex2.simpl}), interconnected with the controllers (\ref{ex2.contr}) via the conditions (\ref{int.constr}). 
\end{enumerate}
\end{proposition}

\begin{proof}
Take {\em any} solution $x^w_\ast$ to (\ref{pre.re}). By definition
\be\label{pre.re.star}
\dot x^w_\ast=f_0(x^w_\ast) +\dst\frac{\sum_{i=1}^n P_i w_i}{n}.
\ee
Observe that such a solution $x^w_\ast$ is necessarily bounded. As a matter of fact, in view of the assumptions on $f_0$, the incremental dissipation inequality (\ref{di.x}) hold and the incremental storage function $V(x^w_\ast, x_\ast)=\frac{1}{2}   (x^w_\ast-x_\ast)^T(x^w_\ast-x_\ast)$ satisfies
$\dot V \le 0$ (in system (\ref{pre.re}) inputs are absent). Hence $V(x^w_\ast, x_\ast)$ is bounded and by regularity of $V$ and boundedness of $x_\ast$, $x^w_\ast$ is bounded. \\
Define now 
\be\label{def.now}
G u^w_i= -\left(P_i w_i-\frac{\sum_{i=1}^n P_i w_i}{n}\right). 
\ee
Observe that since $\sum_{i=1}^n Gu^w_i=\mathbf{0}$ by construction, and $G$ is full-column rank, then $u^w\in \mathcal{R}(B  \otimes I_{p})$, i.e., the requirement imposed by the interconnection condition (\ref{int.constr}) is fulfilled. An explicit expression for $u^w$ can be given. Let 
\[
(I_n\otimes G) u^w = ((\frac{\mathbf{1}_n\mathbf{1}_n^T}{n}-I_n)\otimes I_r) Pw =: (Y\otimes I_r) Pw 
\]
where $r$ is the dimension of the state space of each system. 
Hence (\ref{def.now}) can be rewritten as 
\[
G u^w_i = \dst\sum_{j=1}^n  Y_{ij} P_j w_j,\quad \textrm{with } Y_{ij}=[Y]_{ij}.
\]
There exists a solution $u^w_i$ to the latter equation if and only if $GG^\dag b_i=b_i$ with $b_i=\sum_{j=1}^n  Y_{ij} P_j w_j$, where $G^\dag$ is the Moore-Penrose pseudo inverse. Recalling that $\mathcal{R}(P_j)\subseteq \mathcal{R}(G)$, we can assume the existence of matrices $\Gamma_j$ such that  
\[
P_j w_j =G \Gamma_j w_j. 
\]
As a result 
\[
b_i=\sum_{j=1}^n Y_{ij} P_j w_j=\sum_{j=1}^n  Y_{ij} G \Gamma_j w_j=G \dst\sum_{j=1}^n Y_{ij}  \Gamma_j w_j 
\]
that is $b_i\in \mathcal{R}(G)$, and $GG^\dag b_i=GG^\dag G \dst\sum_{j=1}^n  Y_{ij}  \Gamma_j w_j= G\dst\sum_{j=1}^n Y_{ij}  \Gamma_j w_j= \dst\sum_{j=1}^n Y_{ij}  G \Gamma_j w_j= \dst\sum_{j=1}^n Y_{ij}  P_j w_j= b_i$. Then the unique solution to (\ref{def.now}) is 
$
 u^w_i = G^\dag \dst\sum_{j=1}^N  Y_{ij} P_j w_j.
$
Replacing (\ref{def.now}) into (\ref{pre.re.star}), the latter becomes
\[
\dot x^w_\ast=f_0(x^w_\ast) +G u^w_i+P_i w_i.
\]
The latter  holds true for all $i=1,2,\ldots, n$ thus showing that  
\[
(x^w,u^w)= ((\mathbf{1}_n\otimes I_r) x^w_\ast, ((u^w)^T,\ldots, (u^w_n)^T)^T)
\]
 solves the regulator equations. \\
Bearing in mind that $(B\otimes I_p) \lambda^w =u^w$,   we have  
\[\ba{rcl}
\lambda^w &=& - (B^\dag\otimes I_p)\left(I- \frac{\mathbf{1}_n^T\mathbf{1}_n\otimes I_p}{n}\right) Pw\\
&=&  - \left(B^\dag\left(I- \frac{\mathbf{1}_n^T\mathbf{1}_n}{n}\right)\otimes I_p \right)Pw,
\ea\] 
that is $\lambda^w = H w$.  Using the embedding (\ref{embedding}) with $\tau = Id, \phi=s$ and $\psi(\eta)=H\eta$, and an analogous decomposition as in (\ref{internal.model.k1}), the internal model controller takes the form
\[\ba{rcll}
\dot \eta_k &=& s(\eta_k) \\
\lambda_k &=& H_k^T \eta_k, & k=1,2,\ldots, m.  
\ea\]
The addition of the control term $H_k v_k$
\[\ba{rcll}
\dot \eta_k &=& s(\eta_k) +H_k v_k\\
\lambda_k &=& H_k^T \eta_k, & k=1,2,\ldots, m.  
\ea\]
renders the system incrementally passive, in view of the incrementally passive nature of the map $s(\cdot)$. Recall (see Example \ref{ex2}) that the condition on $F$ that defines the dynamics of the systems according to the identity $f_0=\nabla F$ guarantees incremental passivity of systems (\ref{sys.ex2.simpl}). 
Hence, we are under the conditions of Theorem \ref{th1} and one concludes that the controllers 
\begin{align} \label{eqn.ControllerIdenticalSys}
\begin{split}
\dot \eta_k &= s(\eta_k) -H_k z_k \\
\lambda_k &= H_k^T \eta_k - z_k , \quad k=1,2,\ldots, m.  
\end{split}
\end{align}
with $z=(B^T\otimes I_p) y$ guarantee that the output agreement problem is solved.  \\
\end{proof}

The controllers, designed as in \eqref{eqn.ControllerIdenticalSys}, can be stacked together to the dynamics
\begin{align}\label{stacked.controllers}
\begin{split}
\dot{\eta} &= \bar{s}(\eta) - \bar{H} v \\
\lambda &= \bar{H}\eta - \nu,
\end{split}
\end{align}
where $\eta = [\eta_{1}^{T},\ldots,\eta_{m}^{T}]^{T} \in \mathbb{R}^{mr}$, $\bar{s}(\eta) = [s(\eta_{1})^{T}, \ldots,$ $s(\eta_{m})^{T}]^{T}$, and $\bar{H} = \mathrm{block.diag}\{H_{1}, \ldots, H_{m} \}$.

Bearing in mind the controllers (\ref{stacked.controllers}), 
one conclusion that follows immediately from the proof of the result is that steady state solution of the controllers $\eta^{w}$ can be taken as $\eta^{w} = \1_{n} \otimes w$. That is, one possible steady state solution of the output agreement problem is that each controller dynamics reproduces exactly the disturbance signal. This observation can be used to redesign the controllers. 
In particular, additional communication between the different (distributed) controllers can be used to improve the convergence of the controllers. 

\subsection{Adding communication between controllers}

Consider an additional communication network $\mc{G}_{\mathrm{comm}}$, having one node for each controller, and one edge if the controllers can exchange data.\footnote{for instance, two controllers can exchange data if their corresponding edges are incident to the same node in the original graph $\mathcal{G}$.} 
For simplicity, we assume that $\mc{G}_{\mathrm{comm}}$ is an undirected connected graph. The Laplacian matrix of the communication graph is denoted by $L_{\mathrm{comm}} \in \mathbb{R}^{m\times m}$. 
As we shall see below, the additional communication term allows us to add a diffusive coupling between the various controllers that {\em explicitly} enforces the convergence of all the controllers states $\eta_k$ to the same signal. This in turn guarantees that the stacked vector  $\bar H \bar \eta = {\rm block.diag}\{H_1^T, \ldots, H_m^T\}(\eta_1^T\ldots  \eta_m^T)^T$ converges
to   $H \eta^\ast$, for some $\eta^\ast$. We recall that under the conditions that the convergence to the solution of the output agreement problem is uniform in the initial conditions, such a signal   $\eta^\ast$ must satisfy
$(B\otimes I_p)H \eta^\ast= u^w$.  
If in addition the graph is acyclic and the system $\dot \eta = s(\eta), \lambda = H\eta$ is incrementally observable\footnote{The system $\dot \eta = s(\eta), \lambda = H\eta$ is incrementally observable if any two solutions $\eta, \eta'$ to $\dot \eta = s(\eta)$ which yield the same output necessarily coincide, i.e.~$\eta=\eta'$.}, then necessarily, $\eta^\ast= w$, i.e.~the internal model controllers asymptotically synchronize to the disturbance $w$.

By revisiting now the proof of Theorem \ref{th1}, one can directly see that the assumption of incremental passivity of the controllers, i.e., Assumption \ref{a.ip.im}, is stricter than necessary. In particular, one can  require the incremental passivity property (\ref{di.eta}) not to hold with respect to any two trajectories, but only with respect to the real and the steady state trajectory, i.e., with 
 $\eta_k'=\eta_k^w, v_k'=0, \lambda_k'=\lambda^w$.  Thus, one can replace Assumption \ref{a.ip.im} with the following weaker assumption.  

\begin{myassumption}{3a}
\textit{
Let $\eta^{w}=\tau(w)$ and $\lambda^{w} = \psi(\tau(w))$ be a solution to \eqref{embedding}, and let $v^{w} = 0$. For each $k=1,2,\ldots, m$, there exists regular storage functions $W_k(\eta_k, \eta_k^{w})$, with $W_k:\R^{q_k}\times \R^{q_k}\to \R_+$ such that
\[ \dst\frac{\partial W_k}{\partial \eta_k} \phi_k(\eta_k, v_k) + 
 \dst\frac{\partial W_k}{\partial \eta^{w}_k} \phi_k(\eta^{w}_k, \mathbf{0}) \\
 \le
 (\lambda_k-\lambda_k^{w})^T v_{k}.
\]}
\end{myassumption}
It can be readily seen that the proof of Theorem \ref{th1} remains valid if Assumption \ref{a.ip.im} is replaced by Assumption 3a. In particular, the following result holds. 

\begin{proposition}
Let all assumptions of Proposition \ref{prop.IndenticalSys} hold and let $L_{\mathrm{comm}} \in \mathbb{R}^{m \times m}$ be the Laplacian matrix of communication graph. Then the distributed controller with communication of the form
\begin{align}
\begin{split} \label{eqn.ControllerAugmented}
\dot{\eta} &= \bar{s}(\eta) -(L_{\mathrm{comm}}\otimes I_{r})\eta -  \bar{H} v \\
\lambda &= \bar{H}^{T}\eta 
\end{split}
\end{align}  
 interconnected with the node dynamics (\ref{ex2.contr}) according to (\ref{int.constr}), solves the output agreement problem 
Furthermore, $\lim_{t\to \infty}  ||\eta_k(t)-\eta_j(t)||=0$ for all $k\ne j$. 
\end{proposition}

\begin{proof}
Under the assumptions of  Proposition \ref{prop.IndenticalSys} it holds that 
$\eta^{w} = w$. 
Consequently $(L_{\mathrm{contr}}\otimes I_{r})\bar{\eta}^{w} = 0$. The controller \eqref{eqn.ControllerAugmented} satisfies Assumption 3a, since the directional derivative of $W = \frac{1}{2}(\bar{\eta} - \bar{\eta}^{w})^{T}(\bar{\eta} - \bar{\eta}^{w})$ is
\begin{align*}
\dot{W} \leq - (\eta - \eta^{w})^{T}(L_{\mathrm{comm}}\otimes I_{r}) (\eta - \eta^{w}) + (\lambda - \lambda^{w})^{T}v.
\end{align*}
Mirroring the proof of Theorem \ref{th1}, the derivative of the storage function $U((x,x^{w}),(\eta,\eta^{w}))$ satisfies 
\begin{align*}
\dot{U} \leq -\|B^{T} \otimes I_{p})y\|^{2} - \eta^{T}(L_{\mathrm{comm}} \otimes I_{r})\eta.
\end{align*}
Thus, with the same arguments as in the proof of Theorem \ref{th1}, convergence can be concluded. 
Additionally, this proves that $\lim_{t\to \infty}  \|\eta_k(t)-\eta_j(t)\|=0$ for all $k\ne j$. 
\end{proof}

\section{Relation to known results} \label{sec.Relation}

Next, we compare our results to known results.

\subsection{Static Couplings}

The papers \cite{Stan2007}, \cite{Scardovi2010} study synchronization of cocoercive (or semi-passive(\cite{AP-HN-TCSI01})) systems  that are free from disturbances with purely {\em static} output feedback that uses relative measurements. 
The result extend to systems that have a ``shortage'' of incremental passivity, called  relaxed cocoercive  systems.\footnote{Relaxed cocoercive systems are related to QUAD systems \cite{DeLellis2011}. Any QUAD system, augmented with an additive input on each state and the output being the full state vector, is also relaxed cocoercive.}
For \emph{identical} relaxed cocoercive systems it was shown in \cite{Scardovi2010} (see also \cite{DeLellis2011}) that \emph{static} couplings suffice to ensure synchronization, provided that the network features a sufficiently strong coupling.  \\
Here, \emph{heterogeneous} incrementally passive systems affected by disturbances are considered. 
The heterogeneity of the systems and the time varying external disturbances cause the need for dynamic couplings. However, if all systems already share a common \emph{internal model},  static couplings are also sufficient in our approach.
\begin{proposition}[Static Coupling]\label{prop.StaticCoupling}
Consider the system (\ref{nonl.systems.tot}) and suppose all node dynamics are incrementally passive. 
If there exists a solution to the regulator equations  \eqref{regulator.equations} with $u^{w}(t) = 0$\footnote{This holds in particular if all systems (and exosystems) are identical. More generally, it asserts that all systems incorporate the \emph{same} internal model. },
then, the \emph{static} controller
$
\lambda = \nu
$
with the interconnection $u=(B\otimes I_p) \lambda,$ and $\nu=-(B^T\otimes I_p) y$ solves the output agreement problem. 
\end{proposition}
\begin{proof}
By the incremental passivity property of the subsystems it is true that 
$ \dot{V}  \leq (y - y^{w})^{T}(u-u^{w}),  $
where $V = \sum_{i=1}^{n}V_{i}$. Now, since $u^{w} = 0$ and $y^{w} \in \mc{N}((B \otimes I_{p})^{T})$ the coupling $u = - Ly = - (B\otimes I_{p})(B \otimes I_{p})^{T} y = - (B\otimes I_{p})(B \otimes I_{p})^{T} (y - y^{w})$ gives
$ \dot{V}  \leq (y - y^{w})^{T}L (y - y^{w}).$
Convergence and boundedness can now be shown as in the proof of Theorem \ref{th1}.
\end{proof}

Please note that the input to the systems computes in this case as
\[
u = - (L \otimes I_{p})y,
\]
where $L = BB^{T}$ is the \emph{Laplacian matrix} of the (undirected) graph. Thus, for homogeneous systems, our controller design method reduces to the well-known Laplacian coupling, as studied, e.g., in \cite{Stan2007}, \cite{Scardovi2010}, \cite{DeLellis2011}.
However, it should be remarked that, while incrementally passive systems strictly include the class of cocoercive systems, our results do not appear to be trivially extendable to the class of relaxed cocoercive systems.  Moreover, while the results of \cite{Stan2007}, \cite{Scardovi2010} apply to networked systems over a  balanced, directed graph (possibly even time-varying), our results are given for static undirected graphs.

\subsection{Static Disturbances} \label{sec.StaticDist}

The output agreement problem with \emph{constant} disturbances deserves particular attention. 
Control of passive system with constant disturbances is studied, e.g., in \cite{Jayawardhana2007} or \cite{Hines2011}, where the notion of equilibrium independent passivity is introduced. Equilibrium independent passivity is closely related to incremental passivity, i.e., it is defined by a condition similar to \eqref{di.x} assuming that one of the trajectories (e.g., $x'$) is  an equilibrium trajectory. Optimality properties and a network theoretic interpretation of networks of equilibrium independent passive systems are discussed in \cite{Burger2013}, \cite{Burger2013a}. 
The stability of passive networks with static disturbance signals has also been discussed in \cite{vanderSchaft2012}. We derive here slightly more general\footnote{In fact, differently from \cite{vanderSchaft2012}, \cite{Burger2013}, \cite{Burger2013a}, we do not assume output strict incremental passivity but only incremental passivity.} controllers (or dynamic couplings) as \cite{Burger2013}, \cite{Burger2013a}, using the internal model control approach.

\begin{proposition}

Consider the network $\mc{G}$ with dynamics on the nodes  \eqref{nonl.systems.tot}.
Suppose $w_{i}$ is some constant signal, i.e., $s_{i}(w_{i}) = 0$, the regulator equations \eqref{regulator.equations} hold and \eqref{nonl.systems.tot} are incrementally passive. Then, any controller of the form 
\begin{align}
\begin{split} \label{sys.ControllerStatic}
\dot{\eta}_{k} &= v_{k} \\
\lambda_{k} &= \psi_{k}(\eta_{k}) + \nu_{k}, \; k \in \{1,\ldots,m\}
\end{split}
\end{align} 
with $\psi_{k}(\cdot)$ satisfying the strong monotonicity condition
\begin{align}\label{eqn.Monotonicity}
(\psi_{k}(\eta) - \psi_{k}(\eta'))^{T}(\eta - \eta') \geq c\|\eta - \eta'\|^{2}, \quad \forall \eta, \eta'
\end{align}
for some positive constant $c$, and interconnection constraints \eqref{int.constr} solves the output agreement problem. 
\end{proposition}
Note that the controller \eqref{sys.ControllerStatic} is not necessarily incrementally passive, i.e., Assumption \ref{a.ip.im} is not met. However, we will show next that the controllers satisfy the weaker Assumption 3a. 

\begin{proof}
Let $x^{w}$ and $u^{w}$ be solutions to the regulator equations \eqref{regulator.equations}. By the structure of \eqref{regulator.equations} follows immediately that $v^{w} = -(B\otimes I_{p})^{T}h(x^{w},w) = 0$.
Since the disturbance is static, i.e., $\dot{w} = 0$, the conditions \eqref{embedding}  are solved with $\phi(\cdot) = 0$ and $\tau(w)$ such that
\begin{align} \label{eqn.StaticEmb}
 \lambda_{p}^{w} + \lambda_{0}^{w} = \psi(\tau(w)) 
\end{align}
for some $\lambda_{p}^{w}$ satisfying $u^{w} = (B \otimes I_{p})\lambda_{p}^{w}$ and some $\lambda_{0}^{w}  \in \mc{N}\left( B \otimes I_{p} \right)$  and constant. 
Thus, there is not a unique solution to \eqref{embedding}, but rather for any $\lambda_{0}^{w} \in \mc{N}(B\otimes I_{p})$ there exists exactly one $\tau(w)$ solving \eqref{embedding} (the existence of more than one solution $\tau(w)$ would contradict the strong monotonicity condition (\ref{eqn.Monotonicity})).
Select now $\eta^{w} = \tau(w)$ as a solution to \eqref{eqn.StaticEmb} an \emph{arbitrary} $\lambda_{0}^{w} \in \mc{N}(B \otimes I_{p})$, and let $\lambda^{w} = \lambda_{p}^{w} + \lambda_{0}^{w}$. 
\\
One can construct now for each controller a storage function $W_k$ that satisfies Assumption 3a, i.e., that shows passivity with respect to the constant signals $\lambda_{k}^{w}$ and $v_{k}^{w}=0$.
Let in the following $\Psi_k : \R^q \to \R$ be a twice continuously differentiable function such that $\nabla \Psi_k(\eta_k)= \psi_k(\eta_k)$. Since, by assumption, $\psi_{k}$ satisfy the monotonicity condition \eqref{eqn.Monotonicity} all $\Psi_{k}$ are strongly convex. 
Consider now the following storage function (\cite{Jayawardhana2007}, \cite{Burger2013}):
\begin{align} \label{eqn.Bregman}
\begin{split}
W_k(\eta_k,\eta_k^w)=  \Psi_k(\eta_k)- \Psi_k(\eta_k^w) - 
\nabla \Psi_k^T(\eta_k^w)(\eta_k-\eta_k^w).
\end{split}
\end{align}
Since $\Psi_k$ is convex and, by the global under-estimator property of the gradient, we have
$
\Psi_k(\eta_k)\ge  \Psi_k(\eta_k^w)+
\nabla \Psi_k^T(\eta_k^w)(\eta_k-\eta_k^w)
$
for each $\eta_k, \eta_k^w$. 
 Since $\Psi_k$ is strongly convex, then it is in particular strictly convex and the previous inequality holds if and only if $\eta_k=\eta_k^w$. Then $W_k$ is regular (\cite{Jayawardhana2007}). 
Hence, $W_k$ is a positive regular storage function.
 Furthermore,
\begin{align*}
 \dst\frac{\partial W_k}{\partial \eta_k} \phi_k(\eta_k, v_k) &=
 (\psi_k(\eta_k)-\psi_k(\eta_k^w))^T v_k \\
 &= (\lambda_k-\lambda_k^w)^T v_k.
 \end{align*}
In the case of constant disturbances Assumption 3a is always fulfilled by controllers of the form \eqref{sys.ControllerStatic}.
Mimicking now the proof of Theorem \ref{th1}, using the storage function $U((x,x^{w}),(\eta,\eta^{w})) = V(x,x^{w}) + W(\eta,\eta^{w})$, with $ W(\eta,\eta^{w}) = \sum_{k=1}^{m}W_{k}(\eta_{k},\eta_{k}^{w})$, one obtains
$ \dot{U} \leq -\|(B \otimes I_{p})^{T}y\|^{2}. $
With the same arguments as used in the proof of Theorem \ref{th1}, it follows that the controller \eqref{sys.ControllerStatic} solves the output agreement problem in the case of static disturbances, that is $\lim_{t \rightarrow \infty} (B \otimes I_{p})^{T}y(t) = \0$.

\end{proof}

\section{Time-varying Optimal Distribution Control} \label{sec.OptimalDistribution}

We use now the output agreement theory for the design of (optimal) distribution control laws in distribution networks with storage.

Consider an inventory system with $n$ inventories and $m$ transportation lines, and let $B$ be the incidence matrix of the transportation network. The dynamics of the inventory system is given as
\begin{align}
\begin{split} \label{sys.Inventory1}
\dot{x} &= B\lambda + Pw,
\end{split}
\end{align}
where $x \in \mathbb{R}^{n}$ represents the storage level, $\lambda \in \mathbb{R}^{m}$ the flow along one line, and $Pw$ an external in-/outflow of the inventories, i.e., the supply or demand. 
This basic model is studied, e.g., in \cite{Bauso2006}, \cite{vanderSchaft2012} or in a discrete-time form in \cite{Baric2012}, \cite{Danielson2013}.

We assume here that the exact realization of the supply/demand is unknown, while it is known that it is generated by the dynamics
\begin{align}\label{sys.Inventory_w} \dot{w} = s(w). \end{align}
The \emph{distribution and balancing problem} is to design controllers on the edges of the network, using only measurements of the storage levels of the incident inventories and regulating the flows $\lambda_{k}$ such that instantaneously all possible supply/demand is satisfied and all inventory levels evolve synchronously. 
The balancing problem has been recently studied in \cite{Danielson2013} using predictive control.
By choosing $u = B\lambda$, the problem can be readily formulated as an output agreement problem with time varying disturbance. 

The regulator equations \eqref{eqn.regulator.equations.controller.independent} for the distribution problem are
\begin{align} \label{eqn.FlowRE}
\begin{split}
&\dot{x}^{w} (t)= u^{w}(t) + Pw(t),\\
& u^{w}(t) \in \mc{R}(B), \; x^{w}(t) \in \mc{N}(B^{T}).
\end{split}
\end{align}
The solution to the regulator equation \eqref{eqn.regulator.equations.controller.independent} is
\begin{align}\label{sol.re}
x^{w}(t) = \1_{n}\Bigl(x_{0}^{w} + \int_{0}^{t} \frac{\1_{n}^{T}Pw(s) }{n} ds\Bigr)
\end{align}
where  $x^w_0$ belongs to the projection of $\omega(W\times X\times \Xi)$ onto $\R^{r_1+\ldots+r_n}$.
To see (\ref{sol.re}), note that $u^{w}(t) \in \mc{R}(B)  \Leftrightarrow \1_{n}^{T}u^{w}(t) = 0$. Let now $x^{w}(t) = \1_{n}x^{w}_{*}(t)$, for some $x^{w}_{*}(t) \in \mathbb{R}$. Then, multiplying \eqref{eqn.FlowRE} from the left with the all ones vector gives
$ n\dot{x}^{w} = \1^{T}Pw(t),$ leading to the desired expression.  
The following observation is now a direct consequence. 
\begin{proposition}
The output agreement problem is feasible only if the accumulated imbalance $\bar{w}(t) =  \int_{0}^{t} \frac{\1_{n}^{T}Pw(s) }{n} ds$ is bounded for all $t \geq 0$. 
\end{proposition}
Otherwise the inventory levels (i.e., $x^{w}$) will grow unbounded. 
The corresponding input is naturally given as 
$u^{w}(t) = -\Delta_{n} Pw(t),$
with $\Delta_{n} = (I_{n}- \frac{1}{n}\1_{n}\1_{n}^{T})$, namely the projection of the supply/demand vector to the space orthogonal to $\mathrm{span}\{\1_{n}\}$. 
Next, we verify that the necessary conditions for the output agreement problem are satisfied by showing feasibility of \eqref{embedding}. Note that the controller output must satisfy
$
 B\lambda^{w}=-(I_{n}- \frac{1}{n}\1_{n}\1_{n}^{T})Pw(t).
$

\begin{proposition} \label{prop.FlowFeasibility}
If the network contains a spanning tree then the condition \eqref{embedding} is feasible. 
\end{proposition}
\begin{proof}
Let $\mc{T} \subseteq \mc{G}$ be a spanning tree. Assume without loss of generality that the edges are labeled in such a way that the flow vector can be written as $\lambda^{w} = [\lambda_{\mc{T}}^{wT}, \lambda_{\bar{\mc{T}}}^{wT}]^{T}$, where $\lambda_{\mc{T}}^{w} $ are the flows on the edges in $\mc{T}$ and $\lambda_{\bar{\mc{T}}}^{w}$ are the flows in all other edges. Similarly, the incidence matrix can be represented as $B = [B_{\mc{T}}, B_{\bar{\mc{T}}}]$. 
A feasible flow solution $\lambda_{p}^{w}$ can now be chosen as $\lambda_{\bar{\mc{T}}}^{w} = 0$ and $\lambda_{\mc{T}}^{w}  =-(B_{\mc{T}}^{T}B_{\mc{T}})^{-1}B^{T}_{\mc{T}}Pw(t)$.
Note that $\lambda_{\mc{T}}^{w}$ routes exactly the balanced component of the supply/demand through the network, since $ \lambda_{\mc{T}}^{w}  =-(B_{\mc{T}}^{T}B_{\mc{T}})^{-1}B^{T}_{\mc{T}}(I_{n} - \frac{1}{n}\1_{n}\1_{n}^{T})Pw(t)$ since $B^{T}_{\mc{T}}\1_{n} = 0_{n}$.
Define now 
\[H =  \begin{bmatrix} -(B_{\mc{T}}^{T}B_{\mc{T}})^{-1}B^{T}_{\mc{T}}P \\0 \end{bmatrix},\]
 and note that $\lambda^{w}_{p} = Hw$. Thus,  $\tau = Id$, $\phi(\cdot) = s(\cdot)$ and $\psi$ being the linear function defined by $H$, solve \eqref{embedding}.
\end{proof}

After augmenting the controller with external outputs, a possible routing controller is 
\begin{align}  \label{eqn.RoutingController}
\begin{split}
\dot{\eta} &= s(\eta) + H^{T}v \\
\lambda &= H\eta + \nu.
\end{split}
\end{align}
Note that if $s(\cdot)$ satisfies the standing assumption \eqref{eqn.incrementalProp}, then this controller is incrementally passive.
 
\subsection{Optimal Distribution Control}

We enlarge our control objective and aim to design a feedback controller that achieves an \emph{optimal} routing.
That is, we want to regulate the flows such that they minimize the quadratic cost function 
  \begin{align} \label{eqn.QuadraticCost}
\mc{P}(\lambda) = \frac{1}{2}\lambda^{T}Q\lambda,
\end{align}
with $Q = \mbox{diag}(q_{1},\ldots,q_{m})$ and $q_{k} >0$.  \\
We exploit therefore that the internal model controller achieving balancing is not unique. In particular, we redesign the controller \eqref{eqn.RoutingController} in such a way that it routes the balanced component of the flow through the network in such a way that at each time instant the cost \eqref{eqn.QuadraticCost} is minimized. That is, asymptotically the routing should be such that at each time instant $t$ the following static optimization problem is solved
 \begin{align} \label{prob.OptFlow}
 \min_{\lambda} \mc{P}(\lambda), \quad \mathrm{s.t.} \; \0 = B\lambda + \Delta_{n} Pw,
 \end{align} 
 where $w = w(t)$ is the supply at the respective time. 
 Let now $\zeta \in \mathbb{R}^{n}$ be the multiplier for the equality constraint. The Lagrangian function of \eqref{prob.OptFlow} is
 \begin{align*}
 \mc{L}(\lambda,\zeta) = \frac{1}{2}\lambda^{T}Q\lambda + \zeta^{T}(B\lambda + \Delta_{n} Pw).
 \end{align*}
 One can express the optimality conditions in terms of the dual solution as
 \[
 Q\lambda +B^T \zeta=\mathbf{0},\quad B\lambda+\Delta_n Pw=\mathbf{0},
 \]
from which $B(Q^{-1}B^{T}\zeta) + \Delta_{n} Pw = \mathbf{0}$,
 with the optimal routing being $\lambda = Q^{-1}B^{T}\zeta$. 
Thus the optimal routing/supply pairs are defined as the set
 \[
 \Gamma = \{(\lambda,w)\,:\, Q\lambda\in \mathcal{R}(B^T),\, B\lambda+\Delta_n Pw=\mathbf{0}\}.
 \]
We formalize the optimal distribution problem as follows:
\begin{definition}
The time-varying optimal distribution problem is solvable for the system \eqref{sys.Inventory1}, if there exists a controller \eqref{im.big} such that any solution originating from $\mathcal{W}\times \mathcal{X}\times \Xi$ satisfies (i) $\lim_{t\to \infty} B^T x(t)=\mathbf{0}$ and (ii) $\lim_{t\to \infty} {\rm dist}_\Gamma(\lambda(t), w(t))=0$.    
\end{definition} 
To solve the problem, we proceed in this way. Instead of designing the controller directly for the flows, we design the controllers for the \emph{multipliers}. We take $\tau = Id$ and $\phi(\cdot) = s(\cdot)$ and design a controller of the form
 \begin{align*}
 \begin{split}
 \dot{\eta} &= s(\eta) + H_{v}v \\
 \zeta &= H_{\zeta} \eta,
 \end{split}
 \end{align*}
 where $H_{\zeta}$ and $H_{v}$ are suitable input and output matrices to be designed next. The routing will then be defined as $\lambda(t) = Q^{-1}B^{T}\zeta(t)$.
 For designing $H_{\zeta}$, note that, 
 provided that $v=\mathbf{0}$ and the initial condition $\eta(0)$ is properly chosen, the system above generates the solution $\eta^w(t)=w(t)$. 
Then  $H_{\zeta}$ must be design in such a way that $\zeta^w(t)=H_{\zeta}\eta^w(t)$ satisfies
 the optimality condition
\begin{align*}
BQ^{-1}B^{T} H_{\zeta}\eta^{w}(t) + \Delta_{n}Pw(t) = 0.
\end{align*}
The matrix $L_{Q} = BQ^{-1}B^{T}$ is a weighted Laplacian matrix. As $L_{Q}$ has one eigenvalue at zero, with the corresponding eigenvector $\1$, it is not invertible. 
However, since $\eta^{w}(t) = w(t)$, one possible solution is 
\begin{align} \label{eqn.QuadraticH}
H_{\zeta} = -L_{Q}^{\dagger}P,
\end{align}
where $L_{Q}^{\dagger}$ is the \emph{Moore-Penrose-inverse} of $L_{Q}$, see e.g., \cite{Gutman2004}.
From the properties of $L_Q^\dagger$ follows  
that 
$
B Q^{-1}B^{T}H_{\zeta}\eta^{w} + \Delta_{n}Pw  = - B Q^{-1} B^T L_Q^\dagger P\eta^{w}  + \Delta_{n}Pw  = -\Delta_{n} P \eta^{w} + \Delta_{n}P w = 0$ as desired.
Now, as the controller should be incrementally passive with input $v$ and output $\lambda$, we can design it in the form  
 \eqref{eqn.RoutingController}  taking
 \begin{align} \label{eqn.Hopt}
 H = Q^{-1}B^{T}H_{\zeta}.
 \end{align}
Then, to have incremental passivity, we simply choose $H_{v} = H^{T}$. This choice of the input and output matrix for the controller \eqref{eqn.RoutingController} ensures that the optimal distribution problem is solved.

\begin{proposition}
Consider the inventory system \eqref{sys.Inventory1} with the supply generated by the linear dynamics $\dot{w} = s(w)$, satisfying \eqref{eqn.incrementalProp}.
Consider the controller
\[\ba{rcl}
\dot \eta &=&s(\eta) - {H}^{T} z\\
\lambda  &=& H \eta -z.
\ea\]
with the interconnection condition
$ z = B^{T}x. $
Then, every solution of the closed-loop system is bounded and (i) $\lim_{t\to +\infty} B^{T}x = 0$, and (ii) $\lim_{t\to +\infty} {\dist_{\Gamma}}(\lambda(t), w(t)) = 0$, that is the time-varying optimal distribution problem is solvable.
\end{proposition}

\begin{proof}
First note that the optimal routing $\lambda^w(t)=Q^{-1} B^T \zeta^w(t)$ satisfies the identity
\[\ba{l}
B\lambda^w(t)+Pw(t)=  - B Q^{-1} B^T L_{Q}^{\dagger}P\eta^w(t)+Pw(t)=\\ 
-\Delta_n P\eta^w(t)+Pw(t)=\1 \dst\frac{\1^T P w(t)}{n}.
\ea\]
Since $\1\frac{\1^T P w(t)}{n}=\dot x^w$ (see (\ref{sol.re})), the optimal routing is such that $\dot x^w(t)=B\lambda^w(t)+Pw(t)$.

Now, consider the storage function $U(x-x^{w},\eta,\eta^{w}) = \frac{1}{2}\|x - x^{w}\|^{2} + \frac{1}{2}\|\eta - \eta^{w}\|^{2}$ 
along the solutions of the autonomous system  $\frac{d}{dt} (x-x^w) = B (\lambda-\lambda^w)=-BB^T (x-x^w)+BH\eta-BH\eta^w$, $\dot \eta= s(\eta)-H^T B^T(x-x^w)$, $\dot \eta^w= s(\eta^w)$. 
It  satisfies
\begin{align*}
\dot{U} =& -(x-x^{w})^{T}BB^{T}x+ (x-x^{w})^{T}BH(\eta - \eta^{w}) \\
&+ (\eta - \eta^{w})^{T}(s(\eta) - s(\eta^{w}))  - (\eta - \eta^{w})^{T}H^{T}B^{T}x \\
\leq & -\|B^{T}x\|^{2},
\end{align*}
due to the incremental passivity of the exosystem, i.e., $(\eta - \eta^{w})^{T}(s(\eta) - s(\eta^{w})) \leq 0$. Since $U$ is positive semidefinite and $\eta^w$ is bounded (again by the incremental passivity property of the exosystem), we have that $x-x^w$, $\eta$, $\eta^w$ are all bounded. Then, by LaSalle's invariance principle, the trajectories converge to the largest invariant set such that 
$B^{T}(x-x^w) =B^{T}x = \mathbf{0}$. Thus, there exists $x_{*}$ such that on this set $x-x^w=x_{*}\1$ and the dynamics evolves as
\begin{align}\label{eqn.OptFlowSolInAgg}
\dot{x}_{*}\1 = BH \eta -BH \eta^w,\; \dot{\eta} = s(\eta),\; \dot{\eta}^w = s(\eta^w).
\end{align}
After multiplying by $\frac{1}{n}\1^{T}$  from the left, it follows $\dot{x}_{*} = 0$, proving that $x$ must approach $x^{w}$ modulo a constant.  This proves the claim (i) of the statement.
\\
To prove claim (ii), note that inserting $\dot{x}_{*} = 0$ into \eqref{eqn.OptFlowSolInAgg} and bearing in mind that $\eta^w =w$ gives the necessary condition that in the set where $B^{T}(x-x^w) = \mathbf{0}$ it must hold that
\[\ba{c}
\mathbf{0}=BH \eta -BH \eta^w= BH \eta -BH w= \\
-BQ^{-1}B^{T}L_{Q}^{\dagger}P(\eta -w) = -\Delta_{n}P(\eta - w).
\ea\]
Hence, $\Delta_{n} P \eta = \Delta_{n} P w$. 
The flow on the invariant set is $\lambda = Q^{-1}B^{T}L_{Q}^{\dagger}P\eta$, while the optimal flow is $\lambda^{w} = Q^{-1}B^{T}L_{Q}^{\dagger}Pw$. Together with the previous condition, this implies that there is a vector $\nu \in \mc{N}(B)$ such that $\lambda = \lambda^{w} + \nu$.  We will show next that $\nu$ must be identical to zero. Note that $\nu = \lambda - \lambda^{w}$, and must therefore satisfy
\[\nu = Q^{-1}B^{T}L_{Q}^{\dagger}P(\eta - w). \]
Multiplying the previous equation from the left by $\nu^{T}Q$ leads to
\[ \nu^{T}Q\nu = 0\]
since $\nu^{T}B^{T} = 0$. As $Q$ is by assumption positive definite, the only solution is $\nu = 0$. This proves that in the set where $B^{T}x = 0$ it must hold that $\lambda = \lambda^{w}$, completing the proof.
\end{proof}

 If additionally communication between the controllers is allowed, the controller can be augmented with a consensus term of the form \eqref{eqn.ControllerAugmented}.

\subsection{Simulation Example}
We illustrate the performance of the controller on a design example.
Consider a network with four inventories and five transportation lines as illustrated in Figure \ref{fig.NetworkStructure}.
\begin{figure}
\begin{center}
\begin{tikzpicture}
\node[draw, circle, fill= white, drop shadow] (v1) at (0,1.5) {$x_{1}$};
\node[draw, circle, fill= white, drop shadow] (v2) at (2,1.5) {$x_{2}$};
\node[draw, circle, fill= white, drop shadow] (v3) at (0,0) {$x_{3}$};
\node[draw, circle, fill= white, drop shadow] (v4) at (2,0) {$x_{4}$};

\draw[->] (v1)  -- node[anchor=south]{$\lambda_{1}$} (v2);
\draw[->] (v1)  -- node[anchor=east]{$\lambda_{2}$} (v3);
\draw[->] (v2)  -- node[anchor=south]{$\lambda_{3}$} (v3);
\draw[->] (v2)  -- node[anchor=west]{$\lambda_{4}$} (v4);
\draw[->] (v3)  -- node[anchor=north]{$\lambda_{5}$} (v4);

\draw[<->] (-1,1.5)  -- node[anchor=south]{$w_{1}$} (v1);
\draw[<->] (v2)  -- node[anchor=south]{$w_{2}$} (3,1.5);
\draw[<->] (v3)  -- node[anchor=south]{$w_{3}$} (-1,0);
\draw[<->] (v4)  -- node[anchor=south]{$w_{4}$} (3,0);
\end{tikzpicture}
\caption{Inventory network with four inventories and five transportation lines.}
\label{fig.NetworkStructure}
\end{center}
\end{figure}
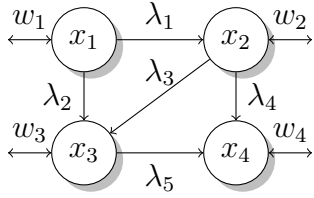
The supply/demand at each inventory is generated by the linear dynamics 
\[
\dot{w}_{i} = \begin{bmatrix}
0 & -s_{i} \\ s_{i} & 0
\end{bmatrix}w_{i} 
\]
 with $s_{1} = 0.1, s_{2} = 0.7, s_{3} = -0.4$ and $s_{4} = -0.2$ and initial conditions $w_{i}(0) = [1,1]^{T}$.  The flow cost function is of the form \eqref{eqn.QuadraticCost} with $Q = \mathrm{diag}\{1,2,3,4,5\}$.
 The controller is implemented in the distributed form with communication \eqref{eqn.ControllerAugmented}, where $H$ is chosen to satisfy \eqref{eqn.Hopt}, and $\mc{G}_{comm}$ is chosen such that two controllers communicate if they are incident to the same inventory. \\
 The simulation results for the inventory levels are shown in Figure \ref{fig.InventoryLevels}. Note that the supply/demand is not balanced, but the accumulated imbalance is bounded. The controller achieves a balancing of the inventory levels.
As an example, the flow $\lambda_{1}(t)$ is shown in Figure \ref{fig.Flow}. The flow approaches fairly quickly the time-varying \emph{optimal flow}.  The simulations illustrate that the controller achieves both objectives, the balancing of the inventory levels and the optimal routing of the flow through the network.
 
\begin{figure}
\begin{subfigure}[t]{.5\linewidth}

  \includegraphics[scale= 0.45]{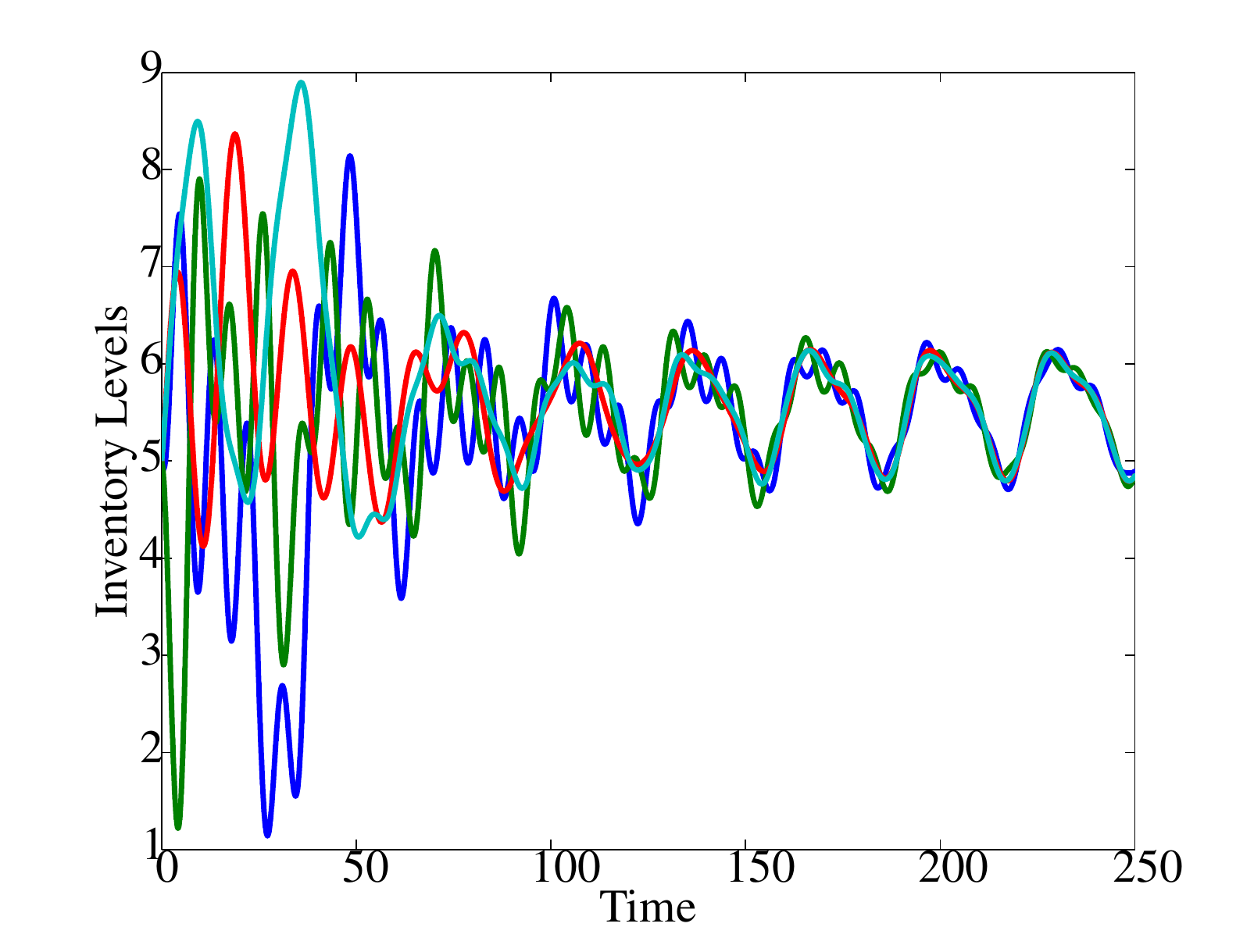}
\caption{Inventory levels $x_{i}(t),\; i \in \{1,2,3,4\},$\\  under time-varying supply/demand.} \label{fig.InventoryLevels}
\end{subfigure}
\begin{subfigure}[t]{.5\linewidth}
\includegraphics[scale= 0.45]{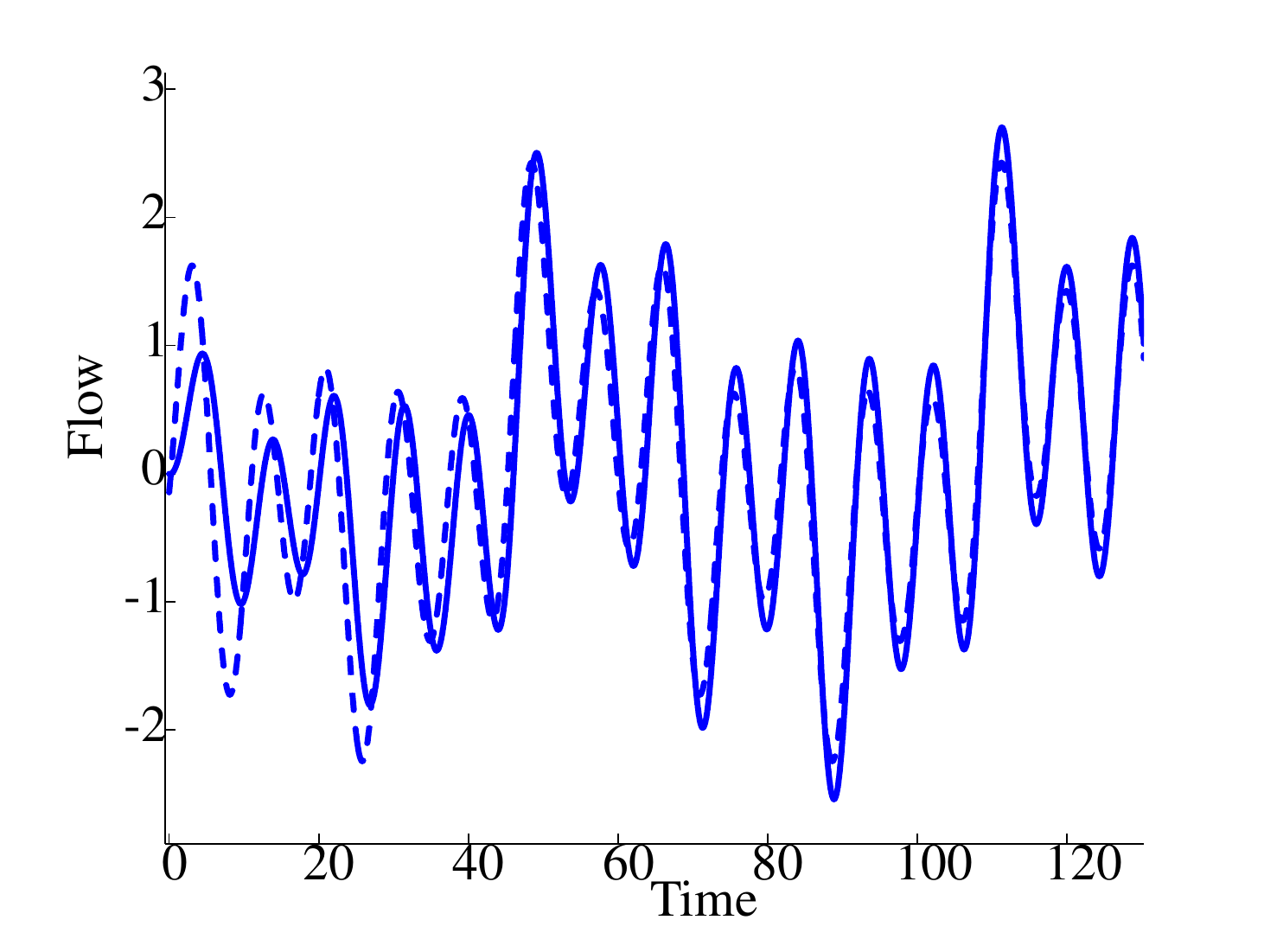}
\caption{Flow $\lambda_{1}(t)$ (solid) and optimal flow on the corresponding edge (dashed). } \label{fig.Flow}
\end{subfigure}
\end{figure}

\section{Power Systems Droop-Control as Internal Model Control} \label{sec.Droop}
In \cite{SimpsonPorco2013} a dynamic oscillatory model of microgrids with frequency-droop controllers is investigated. We provide next an interpretation of the results of \cite{SimpsonPorco2013} in the context of internal model control.

The model of \cite{SimpsonPorco2013} for the frequency-droop controller is
\begin{align} \label{eqn.DroopController}
D_{i} \dot{\theta}_{i} = P_{i}^{*} - P_{e,i}, \; i \in \{1,\ldots,n\},
\end{align}
where  $D_{i}$ is the inverse of the controller gain, $P_{i}^{*}$ is the inverters nominal power, and $P_{e,i}$ is the active electric power. The active electric power is given by
\begin{align} \label{eqn.PowerBalance}
P_{e,i} = \sum_{j=1}^{n} \alpha_{ij} \sin(\theta_{i} - \theta_{j}),
\end{align}
where $\alpha_{ij}$ are constants depending on the node voltages and the line admittance. The coefficients are symmetric $\alpha_{ij} = \alpha_{ji}$ and only non-zero if the two nodes $i$ and $j$ are connected by a line. We refer to \cite{SimpsonPorco2013} for a detailed discussion of the model.
As in \cite{SimpsonPorco2013}, we restrict the discussion in the following to \emph{acyclic} networks. 

In the proposed model, the dynamics in the nodes \eqref{eqn.DroopController} represents the controllers, while the couplings between the nodes, i.e., \eqref{eqn.PowerBalance}, are physical laws. 
Although the situation is reversed to the basic setup of this paper, we can still interpret the droop-controller as an internal model controller.  \\
Consider the node dynamics \eqref{nonl.systems.0} as \eqref{eqn.DroopController}, with node state $x_{i} = \theta_{i}$, constant external signal $w_{i} = P_{i}^{*}$, satisfying $\dot{w}_{i} = 0$,  input $u_{i} = -P_{e,i}$ and output $y_{i} = \dot{\theta}_{i}$, i.e., 
\begin{align} \label{sys.DroopNode}
D_{i} \dot{x}_{i} = P_{i}^{*} + u_{i}, \quad y_{i} = \dot{x}_{i}.
\end{align}
By defining the inputs and outputs in this way, the node dynamics is \emph{output strictly incrementally passive} since for any to inputs $u_{i}, u_{i}'$ and the two corresponding outputs $y_{i}, y_{i}',$ it holds that
\begin{align*}
(y_{i} -y_{i}')(u_{i}- u_{i}')& =(y_{i} -y_{i}')(D_{i}y_{i} - P_{i}^{*} - D_{i}y_{i}' + P_{i}^{*})\\
& = D_{i}\|y_{i} - y_{i}'\|^{2}. 
\end{align*}
From the interpretation of the node dynamics \eqref{nonl.systems.0} as  \eqref{eqn.DroopController}, one notices that $u=-P_{e}= -B A\, \boldsymbol{\sin}(B^T x)$, where $ \boldsymbol{\sin}(z)  = [\sin(z_{1}), \ldots, \sin(z_{m})]^{T}$, $A = \mathrm{diag}\{a_{1},\ldots, a_{m}\}$, $a_k = \alpha_{ij}=\alpha_{ji}$,  and $k$ is the label of the edge connecting nodes $i,j$. One can then interpret the latter equation as the first one of the interconnection conditions (\ref{int.constr}), provided that $\lambda =  A \,\boldsymbol{\sin} (B^T x)$. Set now $\eta= B^T x$. Then 
\begin{align}
\begin{split} \label{sys.DroopEdge}
\dot{\eta} &= -B^{T}y \\
\lambda &= A \, \boldsymbol{\sin}(\eta).
\end{split}
\end{align}
This can be understood as the stacked controllers (\ref{internal.model.k.plus.stab}), where, for all $k$,  $\phi_k(\eta_k, v_k)= v_k$, $v=B^{T}y$, $\psi_k(\eta_k)=a_k \sin(\eta_k)$.

Hence, rewriting the model (\ref{eqn.DroopController})-(\ref{eqn.PowerBalance}) in this way leads directly to an interpretation as an internal model control loop of the form \eqref{sys.ControllerStatic}, where the feed-through term can be omitted, i.e., $\nu = 0$, since the node dynamics is output strictly incrementally passive (see Corollary \ref{coro1}).
We can now restate the result of \cite{SimpsonPorco2013} in the context of internal model control. 
\begin{proposition}\label{prop.droop.controller}
Consider the droop-controller dynamics in the form \eqref{sys.DroopNode} and \eqref{sys.DroopEdge} and let the underlying network $\mc{G}$ be acyclic. Then
\begin{enumerate}
\item the regulator equations \eqref{regulator.equations} are solved by 
$\dot{x}^{w}_{i} = \frac{\sum_{i=1}^{n}P_{i}^{*}}{\sum_{i=1}^{n}D_{i}}=:y^w$ and 
$u^{w}_{i} =D_i\frac{\sum_{i=1}^{n}P_{i}^{*}}{\sum_{i=1}^{n}D_{i}} - P_{i}^{*}$ for all $i=1,2,\ldots, n$;
\item the embedding condition \eqref{embedding}  is feasible if and only if
$ \|A^{-1}(B^{T}B)^{-1}B^{T}u^{w}\|_{\infty} < 1;$
\item if the necessary conditions  \eqref{regulator.equations} and  \eqref{embedding} hold, 
the solutions to the closed loop dynamics \eqref{sys.DroopNode} and \eqref{sys.DroopEdge} with interconnection $u = B\lambda$ and that originate sufficiently close to $x^{w}$ and $\eta^{w}:=\boldsymbol{\sin}^{-1}(A^{-1}(B^{T}B)^{-1}B^{T}u^{w})$ satisfy $\lim_{t \rightarrow\infty} \|y_{i} -y^{w}\| \rightarrow 0$. 
\end{enumerate}
\end{proposition}

The proof follows completely along the lines of the internal model control approach (except the local nature of the stability result) and exploits in particular the results for static disturbances of Section \ref{sec.StaticDist}. For completeness, we provide the proof in the appendix.

 \section{Conclusions}
The paper has investigated output agreement problems in the presence of time-varying disturbances and has discussed the role of dynamic internal-model-based controllers to tackle these problems. We focus on the case in which only relative measurements are available to the controllers and the control applied to the systems  must lie in the range of the incidence matrix. This scenario in fact is very important in distribution networks and is motivated by the physics of the network (Kirchhoff's law). We have examined two of these distribution networks, namely an inventory system and a microgrid, and we have interpreted a load balancing controller and the frequency-droop controller within the proposed framework. Furthermore, in the case of the inventory system, we have shown controllers that achieve an optimal routing. \\
The proposed methodology lends itself to several possible extensions. The use of dynamic controllers could be exploited not only to tackle the presence of exogenous disturbances but also to deal with synchronization problems of heterogenous systems for which a static diffusive coupling does not suffice. In many other distribution networks, and similarly to the inventory system, the constraints imposed by the network induces a non-unique solution to the output agreement problem. It is then meaningful to design controllers that lead to a solution with optimal features. Our approach naturally lends itself to providing such solutions. Other aspects that could be studied are the presence of  uncertainties in the exosystems, larger classes of disturbance signals, robustness to other sources of uncertainties in the dynamical systems. In the current implementation, our internal model controllers depend on all the exosystems generating the disturbance, that could be unfeasible in practice and should be relaxed. 
Moreover, the potentials of our approach in the context of the two case studies have not  been fully explored yet. Phenomena to be studied are for instance the presence of constraints on the input and state variables. For the case of power systems, other classes of controllers could be considered, dealing for instance with the presence of time-varying exogenous inputs.

\appendix

\section{Appendix}

\subsection*{Proof of Proposition \ref{prop.FailureRE}}
\begin{proof}
{\it (Necessity})
If the rank condition in \eqref{eqref.FrancisCond} is violated then there exists a nonzero vector $(x_0, \lambda_0)$ such that 
\be\label{aux0}
 \begin{bmatrix}    \bar{A} - s I_{r} & \bar{G}(B \otimes I_{p}) \\
(B\otimes I_{p})^{T}\bar{C} & \0_{np \times np} \end{bmatrix} 
\begin{bmatrix} 
x_0 \\ \lambda_0
\end{bmatrix}=\mathbf{0}.
\ee
This equality can be made explicit as 
\be\label{aux1}\ba{r}
(\bar{A} - s I_{r}) x_0 + \bar{G}(B \otimes I_{p}) \lambda_0 = \mathbf{0}\\
(B\otimes I_{p})^{T}\bar{C}  x_0= \mathbf{0}.
\ea
\ee
 As $s\not \in \sigma(\bar A)$, then necessarily $\lambda_0\ne \mathbf{0}$.  \\
The graph $\mc{G}$ may or may not contain a cycle. If it does, then statement (1) of the thesis holds. If not, then $(B \otimes I_{p}) \lambda_0\ne \mathbf{0}$ (recall that  $\lambda_0\ne \mathbf{0}$). Since $s$ is not an eigenvalue of any $A_i$, then from (\ref{aux1})
\[\ba{r}
x_0= (s I_{r}-\bar{A})^{-1}\bar{G}(B \otimes I_{p}) \lambda_0\\
(B\otimes I_{p})^{T}\bar{C}  (s I_{r}-\bar{A})^{-1}\bar{G}(B \otimes I_{p}) \lambda_0= \mathbf{0}
\ea\]
The latter shows that $\mc{R}(\bar H(s)(B \otimes I_{p}))\cap \mc{N}((B\otimes I_{p})^{T})\ne \{\mathbf{0}\}$, that is statement (2) of the thesis. \\
{\it (Sufficiency)} If $\mc{G}$ does not contain cycles, then for all $\lambda_0\ne \mathbf{0}$, $(B \otimes I_{p}) \lambda_0\ne  \mathbf{0}$. Consider now the nonzero vector $(x_0, \lambda_0)=(\mathbf{0}, \lambda_0)$. The product 
\[
 \begin{bmatrix}    \bar{A} - s I_{r} & \bar{G}(B \otimes I_{p}) \\
(B\otimes I_{p})^{T}\bar{C} & \0_{np \times np} \end{bmatrix} 
\begin{bmatrix} 
\mathbf{0} \\ \lambda_0
\end{bmatrix}
\]
returns a zero vector. This shows that the null space of the matrix on the left-hand side is singular, that is condition \eqref{eqref.FrancisCond} is violated. \\
If $\mc{R}(\bar H(s)(B \otimes I_{p}))\cap \mc{N}((B\otimes I_{p})^{T})\ne \{\mathbf{0}\}$, then there exists a  vector $\lambda_0\ne \mathbf{0}$ such that  
\[
(B\otimes I_{p})^{T}\bar H(s)(B \otimes I_{p})\lambda_0=\mathbf{0}. 
\]
By definition of $\bar H(s)=\bar{C}  (s I_{r}-\bar{A})^{-1}\bar{G}$, the latter equality becomes
\[
(B\otimes I_{p})^{T}\bar{C}  (s I_{r}-\bar{A})^{-1}\bar{G}(B \otimes I_{p})\lambda_0=\mathbf{0}. 
\]
Define $x_0=(s I_{r}-\bar{A})^{-1}\bar{G}(B \otimes I_{p})\lambda_0$. 
Then 
\[\ba{l}
(B\otimes I_{p})^{T}\bar{C}x_0=\mathbf{0}\\
(\bar{A}-s I_{r})x_0+\bar{G}(B \otimes I_{p})\lambda_0=\mathbf{0},
\ea\]
which is the same as (\ref{aux0}). But this shows that condition \eqref{eqref.FrancisCond} is violated.
\end{proof}

 \subsection*{Proof of Corollary \ref{cor.SPR}}

\begin{proof}
Under the given assumptions, the equations \eqref{eqn.FrancisEquations} are feasible if and only if Condition \ref{LinItem2} in Proposition \ref{prop.FailureRE} does not hold. \\
Suppose, by contradiction, that $H(s)$ is strictly positive real and Condition \ref{LinItem2} in Proposition \ref{prop.FailureRE} holds.
The condition can equivalently be expressed as follows: there exist vectors $v \in \mathbb{R}^{np}$ and $\mathbf{0}\ne \beta \in \mathbb{R}^{p}$ such that
\[ \bar{H}(s) (B \otimes I_{p}) v = \1_{n} \otimes \beta,\; \forall s \in \sigma(\bar{S}).\]
Multiplying the previous condition from the left by $v^{T}(B \otimes I_{p})^{T}$ leads to
\[v^{T}(B \otimes I_{p})^{T} \bar{H}(s) (B \otimes I_{p}) v = 0, \; \forall s \in \sigma(\bar{S}).\]
Since all $s \in \sigma(\bar{S})$ have zero real part, is equivalent to $\tilde{v}^{T}\left( \bar{H}(j\omega ) + \bar{H}(-j\omega ) \right)\tilde{v} = 0$ for all $\tilde{v} = (B \otimes I_{p}) v$ and for some $\omega\in \mathbb{R}$.  \\
This is a contradiction since $H(s)$ being strictly positive real implies that $\left( \bar{H}(j\omega ) + \bar{H}(-j\omega ) \right)$ is positive definite for all $\omega \in \mathbb{R}$.  This proves the statement.
\end{proof} 
 
 \subsection*{Proof of Proposition \ref{prop.droop.controller}}
\begin{proof}
The first statement follows directly after summing all equations \eqref{sys.DroopNode} and noting that there must be a scalar valued function
 $y^{*} \in \mathbb{R}$ such that $y_{i} = \dot{x}_{i}^{w} = y^{*}$ for all $i \in \{1,\ldots,n\}$. \\
To prove the second statement, we choose $\phi(\tau(w)) = 0$ since $\dot{w} = 0$. Now, note that $u^{w} = B\lambda^{w}$, and since the network is acyclic, $\lambda^{w}$ is uniquely defined as
$\lambda^{w} = (B^{T}B)^{-1}B^{T}u^{w}$. Since for an acyclic network $\mc{N}(B) = \{\0\}$, the second condition in \eqref{embedding} becomes
$(B^{T}B)^{-1}B^{T}u^{w} = \psi(\tau(w)),$ where we can take $\psi(\tau(w)) = A\, \boldsymbol{\sin}(\tau(w))$. Thus, we have
\[ \tau(w) = \, \boldsymbol{\sin}^{-1}(A^{-1}(B^{T}B)^{-1}B^{T}u^{w}),\]
which exists if and only if $\|A^{-1}(B^{T}B)^{-1}B^{T}u^{w}\|_{\infty}<1$. \\
To prove local stability we use the standard storage function \eqref{eqn.Bregman}. Note that it can be defined with
$\Psi_{k}(\eta_{k}) =a_k(1 -\cos(\eta_{k}))$.
If the conditions \eqref{embedding} hold, choose $\eta^{w} = \tau(w)$. 
Stability follows now with the storage function $U((x,x^{w}),(\eta,\eta^{w}) = \sum_{i=1}^{n}V_{i}(x_{i},x_{i}^{w}) + \sum_{k=1}^{m}W_{k}(\eta_{k},\eta_{k}^{w})$, with $V_{i}(x_{i},x_{i}^{w}) = 0$ and $W_{k}$ defined as \eqref{eqn.Bregman}. Note that $U$ is positive semidefinite in a neighborhood around $x^{w}$ and $\eta^{w}$ and such that $\eta, \eta^w\in (-\frac{\pi}{2},\frac{\pi}{2})^m$,  and satisfies
\be\label{U.dot} \ba{l}
\dot{U} \leq  -\sum_{i=1}^{n}D_{i}\| y_{i}(t)  - y_{i}^{w}(t) \|^{2}\\
=  -\sum_{i=1}^{n}D_{i}^{-1}\| \sum_{k=1}^m b_{ik} a_k (\sin(\eta_k(t))-\sin(\eta_k^w(t)))\|^{2},
\ea\ee
due to output strict incremental passivity of \eqref{sys.DroopNode}.  
Note that the latter inequality involves only the variables $\eta, \eta^w$. Hence, it shows that the trajectories of the closed-loop system $\dot \eta = -B^T y= -B^T D^{-1} (P^*+BA\boldsymbol{\sin}(\eta))$ are bounded and converge to the set of points where  $BA\boldsymbol{\sin}(\eta)=BA\boldsymbol{\sin}(\eta^w)$ (i.e., to the set of points where $\boldsymbol{\sin}(\eta)=\boldsymbol{\sin}(\eta^w)$, since the graph has no cycles and $A$ is a diagonal matrix) or, equivalently, to the set of points where $y_{i}=y^{w}= \dot x^w_*$ for all $i$.
Thus, any trajectory originating sufficiently close to $x^{w}$ and $\eta^{w}$ satisfies $\lim_{t \rightarrow\infty} \|y_{i} -y^{w}\| \rightarrow 0$. 
\end{proof}

 \stopmodif

{\footnotesize
 \bibliographystyle{alpha}
\bibliography{Literature}
}

\end{document}

%% file: BlockDiagram01.tex
\setlength{\fboxrule}{1.5pt}

\tikzstyle{block} = [draw, fill=white!20, rectangle, 
    minimum height=3em, minimum width=2em]
\tikzstyle{block2} = [draw, fill=white!20, rectangle, 
    minimum height=4em, minimum width=4em]
\tikzstyle{sum} = [draw, fill=white!20, circle, node distance=1cm]
\tikzstyle{input} = [coordinate]
\tikzstyle{output} = [coordinate]
\tikzstyle{pinstyle} = [pin edge={to-,thin,black}]

\begin{tikzpicture}[auto, node distance=3cm,>=latex', line width=1pt]
    \node [input, name=input] {};
    \node [sum, right of=input] (sum) {};
    \node [block, right of=sum, node distance=4cm, minimum height=7em, minimum width=9em] (controller) {
$ \begin{array}[t]{l}
\dot{x}_{i}  = f_{i}(x_{i},u_{i},w_{i}) \\
y_i = h_i(x_i, w_i) \\
\qquad  i=1,2,\ldots, n
\end{array}
$ };
   \node [block, above of=controller, node distance=2cm, minimum width=9em](Plant) {$\dot{w}_{i} = s_{i}(w_{i})$};
    \node [below of=controller, node distance=2cm](center){};
    \node [block2, left of=center, node distance=4cm, minimum height=4em, minimum width=8em] (E) {{\Large $B \otimes I_{p}$}};
    \node [block2, right of =center, node distance=4cm, minimum height=4em, minimum width=8em] (Et){{\Large $(B \otimes I_{p})^{T}$}};
    \node [sum, below of=Et, node distance=1.8cm, label =below:$-$] (sum2) {};
      \node [block, below of= center, node distance=5em, minimum height=7em, minimum width=9em] (edges){
  $ \begin{array}[t]{l}
\dot{\xi}_k = F_{k}(\xi_k, v_k) \\
\lambda_{k} = H_{k}(\xi_k, v_{k}), \\
\qquad k=1,2,\ldots, m,
\end{array} $
      };
 
    \draw [->] (controller) -| node[pos=0.3] {{\Large $y(t)$}} (Et);
    \draw [->] (Et) -- node[pos = 0.6]{{\Large $z(t)$}} (sum2);
     \draw [->] (sum2) -- node[pos = 0.6]{{\Large $v(t)$}} (edges);
    \draw [->] (edges) -|  node[pos=0.3] {{\Large $\lambda(t)$}} (E);
    \draw [->] (E) -- node[pos = 0.98] {}(sum);
    \draw [->] (sum) -- node[pos = 0.5]{{\large $u(t)$}} (controller);
 
 \draw [->] (Plant) --  (controller); 
 
\end{tikzpicture}
